\documentclass[a4paper,english]{lipics-v2021}
\usepackage{graphicx,amssymb,amsmath}
\usepackage{soul}

\usepackage[utf8]{inputenc}
\usepackage{amsthm}
\usepackage[none]{hyphenat} 
\usepackage{hyperref} 
\usepackage{todonotes} 
\usepackage{url} 
\usepackage{cite}
\usepackage{algorithm,algorithmic}
\usepackage{enumitem}

\usepackage{float}

\usetikzlibrary{decorations.pathreplacing}

\newcommand{\remove}[1]{{}}
\newcommand{\anna}[1]{{\color{cyan}Anna says: #1}}

\renewcommand{\comment}[1]{{}}  
\newcommand{\attention}[1]{{}}  
\newcommand{\changed}[1]{{\color{black}#1}}
\newcommand{\therese}[1]{{}}  
\newcommand{\anurag}[1]{{}}  
\newcommand{\graeme}[1]{{}}  
\newcommand{\peter}[1]{{}}  

\newcommand{\calR}{\mathcal{R}}

\newcommand{\eps}{{\varepsilon}}
\let\epsilon=\varepsilon

\nolinenumbers

\title{Distant Representatives for Rectangles in the Plane}

\author{Therese  Biedl}{David R.~Cheriton School of Computer Science, University of Waterloo, Canada}{biedl@uwaterloo.ca}{https://orcid.org/0000-0002-9003-3783}{Supported by NSERC.}
\author{Anna Lubiw}{David R.~Cheriton School of Computer Science, University of Waterloo, Canada}{alubiw@uwaterloo.ca}{}{Supported by NSERC.}
\author{Anurag Murty Naredla}{David R.~Cheriton School of Computer Science, University of Waterloo, Canada}{amnaredla@uwaterloo.ca}{}{}
\author{Peter Dominik Ralbovsky}{David R.~Cheriton School of Computer Science, University of Waterloo, Canada}{peter.ralbovsky@gmail.com}{}{}
\author{Graeme Stroud}{David R.~Cheriton School of Computer Science, University of Waterloo, Canada}{grstroud@uwaterloo.ca}{}{}
\authorrunning{Biedl, Lubiw, Naredla, Ralbovsky, Stroud}

\Copyright{Therese Biedl and Anna Lubiw and Anurag Murty Naredla and Peter Dominik Ralbovsky and Graeme Stroud}

\ccsdesc{Theory of computation~computational geometry}
\keywords{Distant representatives, blocker shapes, matching, approximation algorithm, APX-hardness}
\date{\today}

\EventEditors{Petra Mutzel, Rasmus Pagh, and Grzegorz Herman}
\EventNoEds{3}
\EventLongTitle{29th Annual European Symposium on Algorithms (ESA 2021)}
\EventShortTitle{ESA 2021}
\EventAcronym{ESA}
\EventYear{2021}
\EventDate{September 6--8, 2021}
\EventLocation{Lisbon, Portugal}
\EventLogo{}
\SeriesVolume{204}
\ArticleNo{61}

\begin{document}
\maketitle

\begin{abstract}
The input to the distant representatives problem is a set of $n$ objects in the plane and the goal is to find a representative point from each object while maximizing the distance between the closest pair of points.
When the objects are axis-aligned rectangles, we give 
polynomial time constant-factor approximation algorithms for the $L_1$, $L_2$, and $L_\infty$ distance measures.
We also prove lower bounds on the approximation factors that can be achieved in polynomial time (unless P = NP).
\end{abstract}

\section{Introduction}

The \emph{distant representatives problem} was first introduced by  Fiala et al.~\cite{fiala2005systems}.  The name is a play-on-words 
on the  term ``distinct representatives'' 
from Philip Hall's classic work on bipartite matching~\cite{hall1935representatives}.
The input 
is a set of geometric objects  in  a metric  space.
The goal is to choose one ``representative'' point in each object such  that the points are distant from  each other---more precisely, the objective is to maximize the distance between the  closest pair of representative points.  In the  decision  version of the problem, we are given a bound $\delta$ and  the  question is whether we can  choose one  representative point in each  object such that the distance between any two  points is at least $\delta$.   

The distant representatives problem has   applications to map  labelling and data visualization.
To attach a label to each object, we can find representative points that are at least distance $\delta$ apart, and label each object with a ball of diameter $\delta$  (a square in $L_\infty$) centred at its representative point.  

The distant representatives problem is closely related to
dispersion and packing problems.
When all the objects are copies of a single object, the distant representatives problem becomes the dispersion problem: to  choose $k$ points in a region $R$  to  maximize the  minimum distance between any two chosen points~\cite{baur2001approximation}.  Equivalently, the problem is to pack $k$ disjoint discs (in the chosen metric) of diameter $\delta$ into an expanded region and maximize $\delta$.  
The distant representatives problem is also related to problems of ``imprecise points'' where standard computational geometry problems are solved when each input point is only known to lie within some small region~\cite{loffler2010-box}.

There is a polynomial time algorithm for the distant representatives problem when the objects are segments on a line\changed{\cite{simons1978fast}}.
This result 
comes from 
the scheduling literature---each representative point is regarded as the centre-point of a unit length job.
However, as shown by Fiala et al.~\cite{fiala2005systems}, the decision version of distant representatives becomes NP-hard in 2D when the objects are 
unit discs for the  $L_2$ norm or unit squares for the $L_\infty$ norm.

Cabello~\cite{cabello2007approximation} was the first to consider the optimization version of the distant representatives problem.  He  gave polynomial time approximation algorithms for the cases in 2D where the objects are squares under  the $L_\infty$ norm,  or discs  under the  $L_2$ norm.  The squares/discs may intersect and may have different sizes.  His algorithms achieve an approximation  factor of 2  in $L_\infty$ and $\frac{8}{3}$ in $L_2$, with an  improvement to 2.24 if the input  discs are  disjoint. A main idea in his solution is an ``approximate-placement'' algorithm that  chooses representative points from a fine-enough grid using a matching algorithm; small squares/discs that do not contain grid points are handled separately. 
Cabello noted that the NP-hardness proof of  Fiala et al.~\cite{fiala2005systems} can be modified to prove that there is  no polynomial time approximation scheme (PTAS) for these problems unless P=NP. However, no one has given  exact lower bounds on the approximation factors that can be achieved in polynomial time.

\paragraph*{Our Results}
We consider the distant representatives problem for axis-parallel rectangles in the plane. 
Rectangles are more versatile than squares or circles in many applications, e.g., for labelling rectangular Euler or Venn diagrams~\cite{marshall2005scaled}.

We give polynomial time approximation algorithms to find representative points for the rectangles such that the distance between any two representative points is at least $1/f$ times the optimum.   
The approximation factors $f$ are given in 
Table~\ref{table:bounds} for the $L_1$, $L_2$, and  $L_\infty$ norms. 
Since rectangles are not fat objects~\cite{chan2003polynomial}, Cabello's approach of discretizing the problem by choosing representative points from a grid does not extend.  Instead, we introduce a new technique of ``imprecise discretization'' and 
choose representative points from  1-dimensional shapes (e.g.,~$+$-shapes) arranged in a grid. After that, our plan is similar to Cabello's.  First we solve an approximation version of the decision problem---to find representative points 
so long as the
given distance $\delta$ 
is not too large compared to the optimum $\delta^*$.  
Then we perform a search to find an approximation to 
$\delta^*$.
Unlike previous algorithms which use the real-RAM model, we use the word-RAM model, and thus must address bit complexity issues.  

We accompany these positive results with lower bounds on the approximation factors that can be achieved in polynomial time (assuming P $\ne$ NP).
The lower bounds are shown in Table~\ref{table:bounds}.  They apply even in the special case of horizontal and vertical line segments in the plane.
The results are proved via gap-producing reductions from 
Monotone Rectilinear Planar 3-SAT~\cite{deberg2010optimal}.
These are the first explicit lower bounds on approximation factors for the distant representatives problem for any type of object.

\begin{table}[ht]
    \centering
    \begin{tabular}{c|c c c}
           & $L_1$ & $L_2$ & $L_\infty$\\
           \hline
          upper bound &  5 & $\sqrt{34} \approx 5.83$ & 6\\
          lower bound &  $1.5$ & $1.4425$ & $1.5$ \\
    \end{tabular}
    \caption{Bounds on polynomial time approximation factors for the  distant representatives problem for axis-aligned rectangles  in the plane.  A  lower bound of $x$ means that an approximation factor less than $x$ implies P $=$ NP. 
    (For other $L_p$ norms,
    there are some constant factors, but we have not optimized them.) 
}
    \label{table:bounds}
\end{table}

Finally, we consider the even more special case of unit-length horizontal line segments, and the decision version of distant representatives.  This is even closer to the tractable case of line segments on a line.  However, 
Roeloffzen in his Master's thesis~\cite{roeloffzenfinding} proved NP-hardness for the $L_2$ norm.  We give a more careful proof that takes care of bit complexity issues, and we show that the problem is NP-complete in the $L_1$ and $L_\infty$ norms.

For our algorithms and  our hardness results, we must deal with bit complexity issues. For rectangles under the $L_1$ and $L_\infty$ norms, we show that both the optimum value $\delta^*$ and the coordinates of an optimum solution have polynomially-bounded bit complexity. 
In particular, the decision problems lie in NP.  The $L_2$ norm remains more of a mystery, and the decision problem can only be placed in $\exists \mathbb{R}$ (for an explanation of this class, see~\cite{cardinal2015computational}).


\paragraph*{Background.}
In one dimension, the decision version of the distant representatives problem for intervals on a line was solved by Barbara Simons~\cite{simons1978fast}, as a scheduling problem of placing disjoint unit jobs in given intervals.  To transform the decision version of distant representatives to the scheduling problem,  scale so $\delta = 1$, then expand each interval by $1/2$ on each side. The midpoints of the unit jobs provide the desired solution.
Simons's decision algorithm 
was speeded up to $O(n \log n)$ by Garey et al.~\cite{garey1981scheduling}.
The optimum $\delta^*$ can be found using a binary search---in fact there is a discrete set of $O(n^3)$ possible $\delta^*$ values, which provides an $O(n^3 \log n)$ algorithm. (We see how to improve this to $O(n^2 \log n)$ but we are not aware of any published improvement.)
There has been recent work on the online version of the problem~\cite{chen2019efficient}.  
The (offline) problem is easier when the intervals are disjoint.  More generally, the problem is easier when the ordering of the representative points is specified, or is determined---for example if no interval is contained in another then there is an optimum solution where the ordering of the representative points is the same as the ordering of the interval's left endpoints. This ``dispersion problem for ordered intervals'' can be solved in linear time~\cite{li2018dispersing}. 
In a companion paper to 
this one, we improved this to
a simpler  algorithm using shortest paths in a polygon that solves the harder problem of finding the lexicographic maximum list of distances between successive pairs~\cite{biedl2021dispersion}.

Cabello~\cite{cabello2007approximation} gave polynomial time approximation algorithms for the distant representatives problem for balls in the plane, specifically for squares in $L_\infty$ and for discs in $L_2$, with approximation factors of 2 and $\frac{8}{3}= 2.6\dot{6}$, respectively.  For disjoint discs in $L_2$ he improved the approximation factor to $2.24$.
Jiang  and Dumitrescu~\cite{dumitrescu2012dispersion} further improved the approximation factor for disjoint discs to $1.414$ ($=1/.707$) by adding LP-based techniques to Cabello's  approach. They also considered the case of unit discs, where they gave an algorithm with approximation factor
 $2.14$ ($= 1/.4674$).
 For disjoint unit discs they gave a very simple algorithm with approximation factor $1.96$ ($= 1/.511$).
In a follow-up paper Jiang  and Dumitrescu~\cite{dumitrescu2015systems} gave bounds on the optimum $\delta^*$ for balls and cubes in $L_2$ depending on the minimum area of the union of subsets of $k$ objects---these results have the flavour of Hall's classic condition for the existence of a set of distinct representatives.

The geometric dispersion problem (when all objects are copies of one object) was studied by Bauer and Fekete~\cite{baur2001approximation}. 
They considered the problem of 
placing $k$ points  in  a rectilinear polygon with  holes to  maximize the min $L_\infty$  distance between any two  points or between a point and the boundary of the region.  Equivalently, the problem is to pack $k$ as-large-as-possible identical squares into the region. 
They gave a polynomial time $3/2$-approximation algorithm, and proved that $14/13$ is a lower bound on the approximation factor achievable in polynomial time.  
By contrast, if the goal is to pack as many squares of a given size into a region, 
the famous shifting-grid strategy of 
Hochbaum and Maas~\cite{hochbaum1985approximation} 
provides a PTAS.  Bauer and Fekete use this PTAS to design an approximate decision algorithm for their problem.

It is NP-hard to decide whether a square can be packed with given (different sized) 
squares~\cite{leung1990packing} or discs~\cite{demaine2010circle}.
For algorithmic approaches, see the survey~\cite{hifi2009literature}.
There is a vast literature on the densest packing of equal discs/squares in a region (e.g.~a large circle or square)---see the book~\cite{szabo2007new}.

Many geometric packing problems suffer from issues of bit complexity.  In particular, there are many packing problems that are not known to lie in NP (e.g.,~packing discs in a square~\cite{demaine2010circle}).
This issue is addressed in 
a recent general approach to geometric  approximation~\cite{erickson2020smoothing}.  
Another direction is to prove that packing problems are complete for the larger class $\exists \mathbb{R}$ (existential theory of the reals)~\cite{abrahamsen2020framework}.

The distant representatives problem is  closely  related to problems on imprecise points, where each point is only known to lie within some $\epsilon$-ball, and the worst-case or best-case representative points, under various measures, are considered.  Many geometric  problems on points (e.g., convex hulls, spanning  trees) have been explored  under the model  of imprecise points~\cite{chambers2017connectivity,dorrigiv2015minimum,loffler2010-box,loffler2010-CH}.

As mentioned above, the distant representatives problem has application to labelling and visualization, specifically it provides a new approach to the problem of labelling (overlapping) rectangular regions or line segments. 
Most map labelling research is about labelling point features with rectangular labels of a given size, and the objective is to label as many of the points as possible~\cite{formann1991packing}. There is a small body of literature on labelling line features~\cite{doddi1997map,wolff2001simple}, and even less on labelling regions, except by assuming a finite pre-specified set of label positions~\cite{wagner2001three}.


\paragraph*{Definitions and Preliminaries}
Suppose we are given a set $\calR$ of $n$ axis-aligned rectangles in 2D.
In the {\em distant representatives
problem}, the goal is to 
choose a point  $p(R) \in R$ for each rectangle $R$ in  $\cal R$ so as to 
maximize the minimum pairwise distance between points, i.e., we want to 
maximize $\min_{R,R'\in \calR} d_\ell(p(R),p(R'))$, where $d_\ell$ is the distance-function
of our choice.  We consider here $\ell=1,2,\infty$, i.e., the $L_1$-distance, the
Euclidean $L_2$-distance and the $L_\infty$-distance. 
We write $\delta^*_\ell$ for
the maximum such distance for $\ell\in \{1,2,\infty\}$, and omit `$\ell$' when it is
clear from the context.

In the {\em decision version} of the distant representatives problem, we are given
not only the rectangles but also a value $\delta$, and we ask whether there exists
a set of representative points that have pairwise distances at least $\delta$.

\section{Approximating the decision problem}
\label{sec:decision-alg}

In this  section we give an algorithm that takes as input a set $\cal R$ of axis-aligned rectangles, and a value $\delta$ 
and finds a set of representative points of distance at least $\delta$ apart so  long as $\delta$ is at most some fraction of the optimum, $\delta^*$, for this instance.  
Let $n = |{\cal R}|$ and suppose that the coordinates of the rectangle corners are even integers in the range $[0,D]$ (which guarantees that the rectangle centres also have integer coordinates).

The idea of the algorithm is to overlay a grid of \emph{blocker-shapes} on top of the rectangles as shown in Figure~\ref{fig:blocker-shapes}, while ensuring that any two blocker-shapes are distance at least $\delta$ apart.  
The hope is to use a matching algorithm to match every rectangle to a unique intersecting blocker-shape.  Then, if rectangle $R$ is matched to blocker-shape $B$, we choose any point in $R \cap B$ as the representative point for $B$, which guarantees distance at least $\delta$ between representative points since the blocker-shapes are distance at least $\delta$ apart.
The flaw in this plan is that there may 
be \emph{small} rectangles that do not intersect a blocker shape.  To remedy this, we represent a small rectangle by its centre point, and we eliminate any nearby blocker-shapes before running the matching algorithm.

For the $L_1$ and $L_\infty$ norms 
we assume that $\delta$ is given as a rational number with at most $t$ digits in the numerator and denominator.  
Because we are using the word-RAM model where we cannot compute square roots, we will work with $\delta^2$ for the $L_2$ norm.
Thus, for the $L_2$ norm, we assume that we are given $\delta^2$ as a rational number with at most $t$ digits in the numerator and denominator.
The bit size of the input is $\Theta(n\log D + t)$. 
Similarly, for $L_2$, any output representative point $(x,y)$ will be given as $(x^2,y^2)$.
With these nuances of input and output, we express the main result of this section as follows:

\begin{theorem}
\label{thm:placement}
There exists an algorithm {\sc Placement}$(\delta)$ that, given input
$\ell\in \{1,2,\infty\}$, 
rectangles $\cal R$, and
$\delta>0$, 
either finds an assignment of representative points for $\cal R$ of $L_\ell$-distance at least 
$\delta$, or determines that 
$\delta > \delta^*_\ell / f_\ell$.
Here
$f_1 = 5, f_2 = \sqrt{34} \approx 5.83, f_\infty = 6$.

The run-time of the algorithm is $O(n^2 \log n)$  in the word RAM model, i.e., assuming we can do basic arithmetic on numbers of size $O(\log D + t)$ in constant time.
\end{theorem}

To describe our algorithm, we think of
overlaying the $D \times D$ bounding box of the rectangles 
with a 
grid of horizontal and  vertical lines such that the diagonal distance across a square of the grid is $\delta$.  This means that grid lines are 
spaced $\gamma_\ell$  apart, where
$\gamma_1 = \delta /2$, 
$\gamma_2 = \delta / \sqrt{2}$,
and $\gamma_\infty  = \delta$.
For $L_2$ we will work with $\gamma_2^2=\delta^2/2$ which is rational.
Note that the algorithm does not explicitly construct the grid.
Number the grid lines from left to right and bottom to top, and identify a grid point by its two indices.
Note that the number of indices is $D/\gamma_\ell$, so the 
size of each index is $O(\log D + t)$.
We imagine filling the grid with {\em blocker-shapes}, where the chosen shape depends on 
the norm $L_\ell$ that is used---see Figure~\ref{fig:blocker-shapes}.

\begin{itemize}
\itemsep -1pt
\item For $\ell=1,2$, we use {\em $+$-shapes}.  Each $+$-shape consists
of the four incident grid-segments of one \emph{anchor} grid-point, where
$(i,j)$ is the anchor of a $+$-shape iff $i$ is even and $i \equiv j \mod 4$.

\item For $\ell=\infty$, we use {\em $L$-shapes}.   
Each  $L$-shape consists of 
the two incident grid-segments above and to the right 
of one \emph{anchor} grid-point,
where $(i,j)$ is the anchor of an $L$-shape iff $i \equiv j \mod 3$.

\end{itemize}
Observe that, by our choice of grid size $\gamma_\ell$,  any two blocker shapes are distance $\delta$ or more apart in the relevant norm.

\begin{figure}[ht]
\hspace*{\fill}
\includegraphics[scale=0.9,page=1,trim=10 30 90 80,clip]{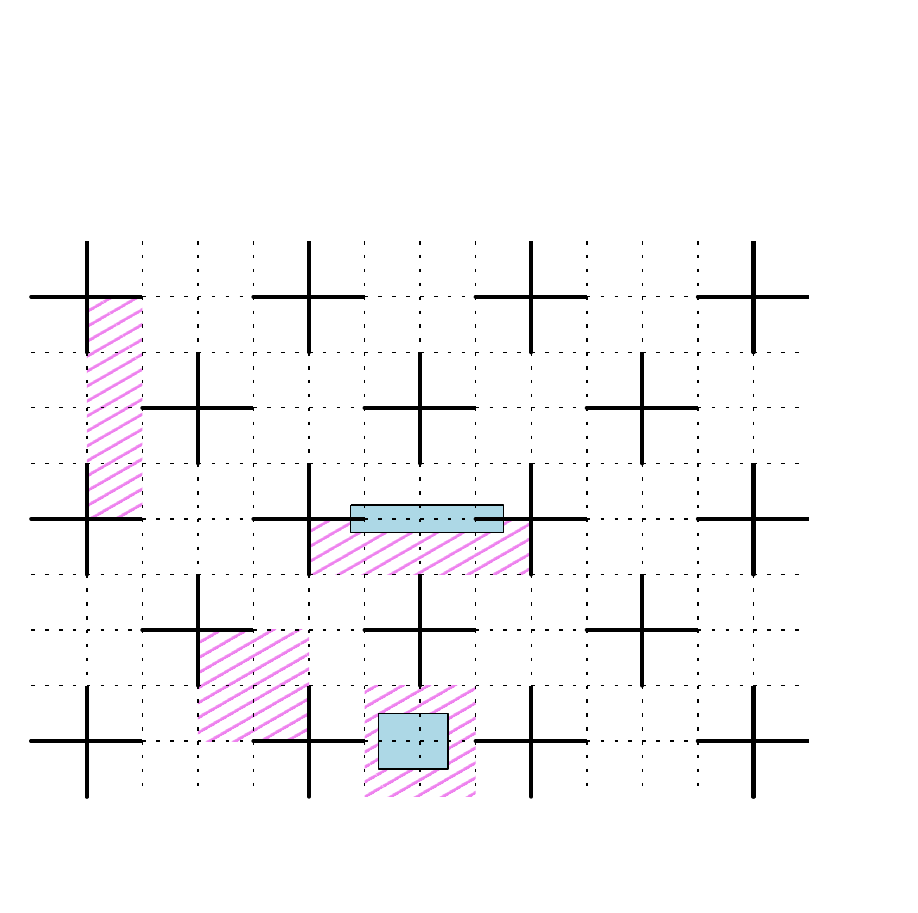}
\hspace*{\fill}
\includegraphics[scale=0.9,page=1,trim=60 30 30 10,clip]{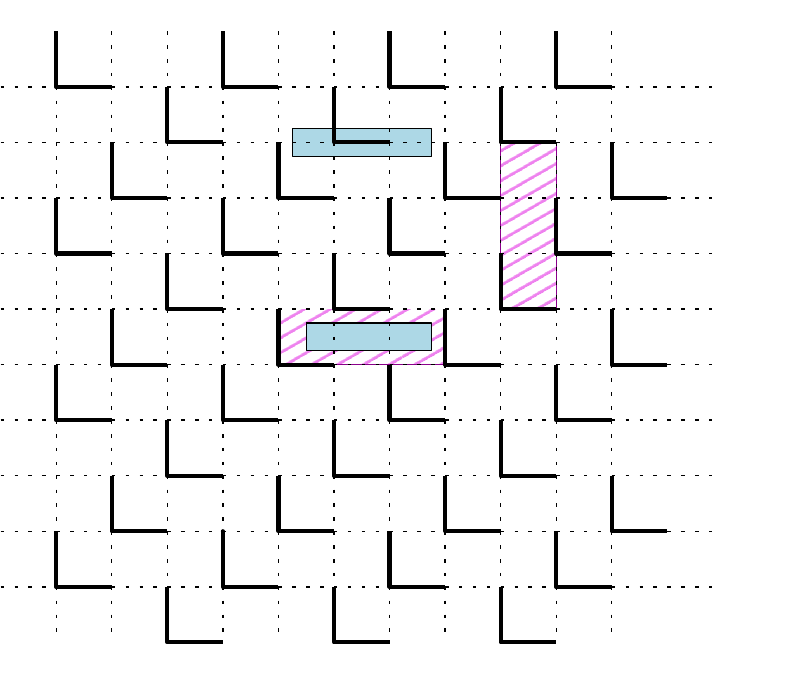}
\hspace*{\fill}
\caption{Grids and blocker-shapes.  We indicate in each a big and a small rectangle (shaded blue) and some 
cavities
(hatched pink).  Small rectangles are contained in cavities.
(Left) The grid of $+$-shapes
for the $L_1$ and $L_2$ norms,
with two long cavities and two square cavities.
(Right) The grid of $L$-shapes
for the $L_\infty$-norm with two long cavities.
}
\label{fig:blocker-shapes}
\label{fig:cells}
\end{figure}

\paragraph*{\bf Algorithm {\sc Placement}($\delta$).}
We now give the rough outline of 
our algorithm
to compute a representative point $p(R)$ for each rectangle $R$.  
The details of how to implement each step are given later on. 
Our algorithm consists of the following steps:
\begin{enumerate}
\item Partition the input rectangles into \emph{small} and \emph{big} rectangles.
Roughly speaking, a rectangle is \emph{big} if it intersects a blocker-shape, but we give a more precise definition below to deal with  intersections on the boundary of the rectangle.  

\item For any small rectangle $r$, let $p(r)$ be the centre of $r$, i.e., the point where the two diagonals of $r$ intersect each other. %
\item 
If two points $p(r),p(r')$ of two small rectangles $r,r'$ have $L_\ell$-distance less than $\delta$,  then declare that $\delta^*< f_\ell\delta$, and halt.

\item 
Find all the blocker-shapes that are \emph{owned} by small rectangles, where 
a blocker-shape $B$ \emph{is  owned by} a small rectangle $r$ if
$p(r)$ has distance strictly less than $\delta$ to some point of $B$.
For $L_2$ we will enlarge ownership as follows: $B$ \emph{is  owned by} $r$ if
$d_1(p(r), B) < \sqrt{2} \delta$.
To justify that this enlarges ownership, note that $d_1(p, q) \le \sqrt{2} d_2(p,q)$ so $d_2(p(r),B) < \delta$ implies $d_1(p(r), B) < \sqrt{2} \delta$.

\item Define a bipartite graph $H$ as follows. On one side, $H$ has a vertex for each big rectangle, and on the other side, it has a vertex for each blocker-shape that is not owned  by a small rectangle.  Add an edge whenever the rectangle intersects the blocker-shape.
\item Construct a subgraph $H^-$ of $H$ as follows.  For any big rectangle, if it has degree more than $n$ in $H$, then arbitrarily delete incident edges until it has degree $n$. 
Also delete any blocker-shape that has no incident edges.  
\item Compute a maximum matching $M$ in $H^-$.    We say that it {\em covers all big rectangles} if every big rectangle has an incident matching-edge in $M$.
\item 
If $M$ does not cover all big rectangles, then declare that $\delta^*<  f_\ell \delta$ and  halt.
\item 
For each big rectangle $R$ let $B$ be the blocker-shape that $R$ is matched to, and let $p(R)$ be an arbitrary point in $B\cap R$.
(This exists since $(B,R)$ was an edge.)
\item Return the set $\{p(R)\}$ (for both big and small rectangles $R$) as an approximate set of distant representatives.
\end{enumerate}

We now define \emph{big} rectangles more precisely.
The intuition is that a rectangle is big if it intersects a blocker shape even if $\delta$ is decreased by an infinitesimal amount.
Note that, as $\delta$ decreases, the blocker-shapes change position and size continuously.  More formally, a rectangle is \emph{big} 
with respect to $\delta$ if 
there is some $\epsilon_0 > 0$ such that for all $\epsilon$, $0 \le \epsilon < \epsilon_0$, there is a point in the (closed) rectangle and in a blocker-shape (for the blocker-shapes at $\delta - \epsilon$).
The reason for this definition is so the set of big rectangles remains the same if $\delta$ is decreased by an infinitesimal amount, a property that becomes relevant when we use the {\sc Placement} algorithm to approximately solve the optimization version of distant representatives.

For implementation details and the correctness proof, we need one more definition. 
A \emph{cavity} is a closed maximal axis-aligned rectangular region  with no points of blocker shapes in its interior.  We distinguish a 
\emph{square cavity}, which is a $2 \times 2$ block of grid squares (only possible for $+$-shapes), 
and a 
\emph{long cavity} which lies between two consecutive grid lines.  For $+$-shapes a long cavity is a $1 \times 4$ or $4 \times 1$ block of grid squares, and for $L$-shapes, a long cavity is a $1 \times 3$ or $3 \times 1$ block of grid squares.
Observe that any small rectangle is contained in a cavity.


\subparagraph*{Implementation and Runtime.}
In order to implement the algorithm efficiently we discuss:  
\begin{itemize}
\item How to test whether a rectangle is big/small.
\item How to find the blocker shapes owned by a small rectangle.
\item How to construct $H^-$.
\item How to efficiently compute the matching.
\end{itemize}

We first show how to find which grid square contains a given point.
Identify a grid square by the indices of its lower left grid point.
Given a point $(x,y)$ in the plane (e.g., a corner of an input rectangle) the vertical grid line just before $x$ has index $i$ where 
$i\gamma_\ell \le x < (i+1)\gamma_\ell$
so $i=\lfloor x/\gamma_\ell\rfloor$.
For $\ell=1$,   $i  
=\lfloor 2x/\delta \rfloor$.  
For $\ell=\infty$, $i = \lfloor x/\delta \rfloor$.
For $\ell=2$,  $i$ is the largest natural number such that $i^2 \le 2x^2/\delta^2$, i.e.,~$i$ is the integer square root of $\lfloor 2x^2/\delta^2 \rfloor$.  The integer square root of a number with $O(\log D + t)$ bits can be found in time 
$O(\log D + t)$ on a word RAM.

We apply the above procedure $O(n)$ times to find the grid squares of all the rectangles' corners and centres.
Using the 
differences and parities of the indices of the grid squares containing the corners, 
we can test if a rectangle 
contains points of blocker shapes in its interior or on its boundary.
From this, we 
can test if a rectangle is big or small in constant time. (Note that our complicated rule is really just testing boundary conditions.)

Each small rectangle $r$ owns 
a constant number of
blocker shapes and these can be found by testing 
a constant number of grid squares
that are near $p(r)$.

Next we show how to construct the bipartite graph $H^-$ and compute a maximum matching.  Note that blocker-shapes, which form one vertex set of $H^-$, are specified using  
$O(\log D + t)$ bits each, although we do not write that in our run-time bounds.
To construct $H^-$ we first build a dictionary for the $O(n)$ blocker-shapes owned by small rectangles.
Then for each big rectangle $R$, enumerate 
blocker-shapes intersecting $R$ in arbitrary order until we have found $n$ that are not owned  by a small rectangle, or until we have found all of them, whichever happens first.  
The run-time for this step is $O(n^2 \log n)$
which will in fact be the bottleneck in our runtime.
The graph $H^-$ has $O(n^2)$ vertices and edges.

To find the maximum matching in $H^-$, we can use the standard algorithm by Hopcroft and Karp \cite{Hopcroft73Karp}
which has run-time $O(\sqrt{\nu}|E|)$, where $\nu$ is the size of the maximum matching 
\cite[Theorem 16.5]{schrijver2003combinatorial}.  
We have $\nu\leq n$ and $|E|=O(n^2)$, so the
run-time to find the matching is $O(n^{2.5})$.
With appropriate further data structures the runtime of computing the matching can be reduced to $O(n\sqrt{n}\log n)$; see the full version.
Therefore the runtime becomes $O(n^2 \log n)$.


\paragraph*{Correctness.}
The algorithm outputs either a set of points or a 
 declaration that
$\delta^*_\ell < f_\ell \delta$.
We first show that the algorithm is correct if it outputs a set of points.

\begin{lemma}
If the algorithm returns a point-set, then
the $L_\ell$-distance between  any two points  chosen by the algorithm 
is at least $\delta$.
\end{lemma}
\begin{proof}
For two small rectangles $r,r'$, this holds since we test $d_\ell(p(r),p(r'))$
explicitly.  
For any two big rectangles $R,R'$, the two assigned points $p(R)$ and $p(R')$ 
lie on different blocker-shapes, and hence have  distance 
at least $\delta$.  
For any big rectangle $R$ and small rectangle $r$, point $p(R)$ lies on 
a blocker-shape that 
is  not owned by  $r$, so the blocker shape, and hence $p(R)$, has
distance at least $\delta$
from $p(r)$.  
\end{proof}

If the algorithm does not output a set of points, then it 
outputs a declaration that $\delta$ is too large compared to the optimum $\delta^*$, viz., 
$\delta^*_\ell < f_\ell \delta$.  This declaration is made either in Step 3 because the points chosen for small rectangles are too close, or in Step 8 because no matching is found. 
We must prove correctness in each case, Lemma~\ref{lem:approxDistant} for Step 3, and Lemma~\ref{lem:approxSucceeds} for Step 8.
In the remainder of this section we 
let $p^*(R)$, $R \in {\cal R}$ denote an optimum set of distant representatives, i.e., $p^*(R)$ is a point in $R$ and every two such points have $L_\ell$-distance at least $\delta_\ell^*$.

\begin{lemma}
\label{lem:approxDistant}
If two points $p(r),p(r')$ of two small rectangles $r,r'$ have distance
less than $\delta$, then $\delta^*_\ell < f_\ell \delta$.
\end{lemma}
\begin{proof}
We first show that 
for any small rectangle $r$, points  $p^*(r)$ and $p(r)$ are close together, specifically, $d_\ell(p^*(r),p(r))\leq 2.5\gamma_\ell$.
Because $L_1$-distance dominates $L_2$ and $L_\infty$-distances, 
it suffices to prove that 
$d_1(p^*(r),p(r))\leq 2.5\gamma_\ell$.
Any small rectangle is contained in a cavity.   
The $L_1$ diameter of a cavity (i.e., the maximum distance between any two points in the cavity) is at most $5 \gamma_\ell$---it is $5 \gamma_\ell$ for a long cavity with $+$-shapes; $4 \gamma_\ell$ for a square cavity with $+$-shapes; and $4 \gamma_\ell$ for a long cavity with L-shapes.  
This implies that any point of $r$ is within distance $2.5 \gamma_\ell$ from $p(r)$, the centre of rectangle $r$.  

Now consider two small rectangles $r$ and $r'$ with $d_\ell(p(r), p(r')) < \delta$.  We will bound the distance between $p^*(r)$ and $p^*(r')$ by applying the triangle inequality: 
\[d_\ell(p^*(r),p^*(r'))
\leq d_\ell(p^*(r),p(r)) + d_\ell(p(r),p(r')) + d_\ell(p(r'),p^*(r')) 
< 2.5\gamma_\ell + \delta +2.5\gamma_\ell =  \delta +  5 \gamma_\ell .
\]
Plugging in  the values 
$\gamma_1 = \delta /2$, 
$\gamma_2 = \delta / \sqrt{2}$,
and $\gamma_\infty  = \delta$,
we obtain bounds of $3.5\delta$, $(1 + 5/ \sqrt{2})\delta \approx 4.5\delta$, and  $6\delta$, respectively.  
Since $f_1 = 5$, $f_2 \approx 5.8$, and $f_\infty = 6$, 
these bounds are at most $f_\ell \delta$ in  all three cases.
Thus $\delta^* < f_\ell \delta$, as required.
\end{proof}

\begin{lemma}
\label{lem:approxSucceeds}
If there is no matching $M$ in $H^-$ that covers all big rectangles,
then $\delta^*_\ell < f_\ell \delta$.
\end{lemma}
\begin{proof}
We prove the contrapositive, using the following plan.
Take an optimal set of distant representatives, $p^*(R)$, $R \in {\cal R}$ with $L_\ell$-distance $\delta_\ell^*\geq f_\ell\delta$. 
For any big rectangle $R$, we ``round'' $p^*(R)$ to a point $b(R)$ that is in $R$ and on a blocker-shape $B(R)$.  
More precisely, we define $b(R)$ to be a point 
that is in $R$, on a blocker-shape, and closest (in  $L_\ell$ distance) to $p^*(R)$.  
In case  of ties, choose  $b(R)$ so  that the smallest rectangle containing $p^*(R)$ and $b(R)$ is minimal (this is only relevant in $L_\infty$).  Break further ties arbitrarily.
Observe that $b(R)$ exists, since a big \changed{rectangle} contains blocker-shape points. Define $B(R)$ to be the blocker-shape containing $b(R)$. \therese{Rearranged a lot to save space.}

By Lemma~\ref{lem:matchH} (stated below) the pairs $R, B(R)$ form a matching in $H$ that covers all big rectangles.  
We convert this to a matching in $H^-$ by repeatedly applying the following exchange step.
If big rectangle
$R$ is matched to a blocker shape $B(R)$ that is not in $H^-$, 
then $R$ has degree exactly
$n$ in $H^-$.  Not all its $n$ neighbours can be used in the current matching since there are at most $n-1$
big rectangles other than $R$.  So change the matching-edge at $R$ 
to go to one of the unmatched neighbours in $H^-$ instead.
\end{proof}

\begin{lemma}  
\label{lem:matchH}
Let $R$ be a big rectangle and let $B = B(R)$.
If $\delta^*\geq f_\ell\delta$ then (1) no other big rectangle $R'$ has $B(R') = B$, and (2) no small rectangle owns $B$.
\end{lemma}

The general idea to prove this lemma
is to show that either type of collision ($B(R) = B(R')$ or $B(R)$ owned by $r$) gives points $p^*$ that are ``close together'', where close together means in a ball of appropriate radius centred at the anchor of $B(R)$.

Let $C_\ell(B)$ be the open $L_\ell$-ball centred at the anchor of $B$ and with diameter $f_\ell \delta$.  
See Figure~\ref{fig:blocker-ball} and note that the diameter of the ball in the appropriate metric is: 
\[f_1 \delta = 5 \delta = 10 (\delta /2)  = 10 \gamma_1
\quad
\quad
f_2 \delta = \sqrt{34} \delta = \sqrt{68} (\delta / \sqrt{2}) = \sqrt{68} \gamma_2
\quad
\quad
f_\infty \delta = 6 \delta = 6\gamma_\infty
\]

\begin{figure}[htb]
\hspace*{\fill}
{\includegraphics[width=.3\textwidth]{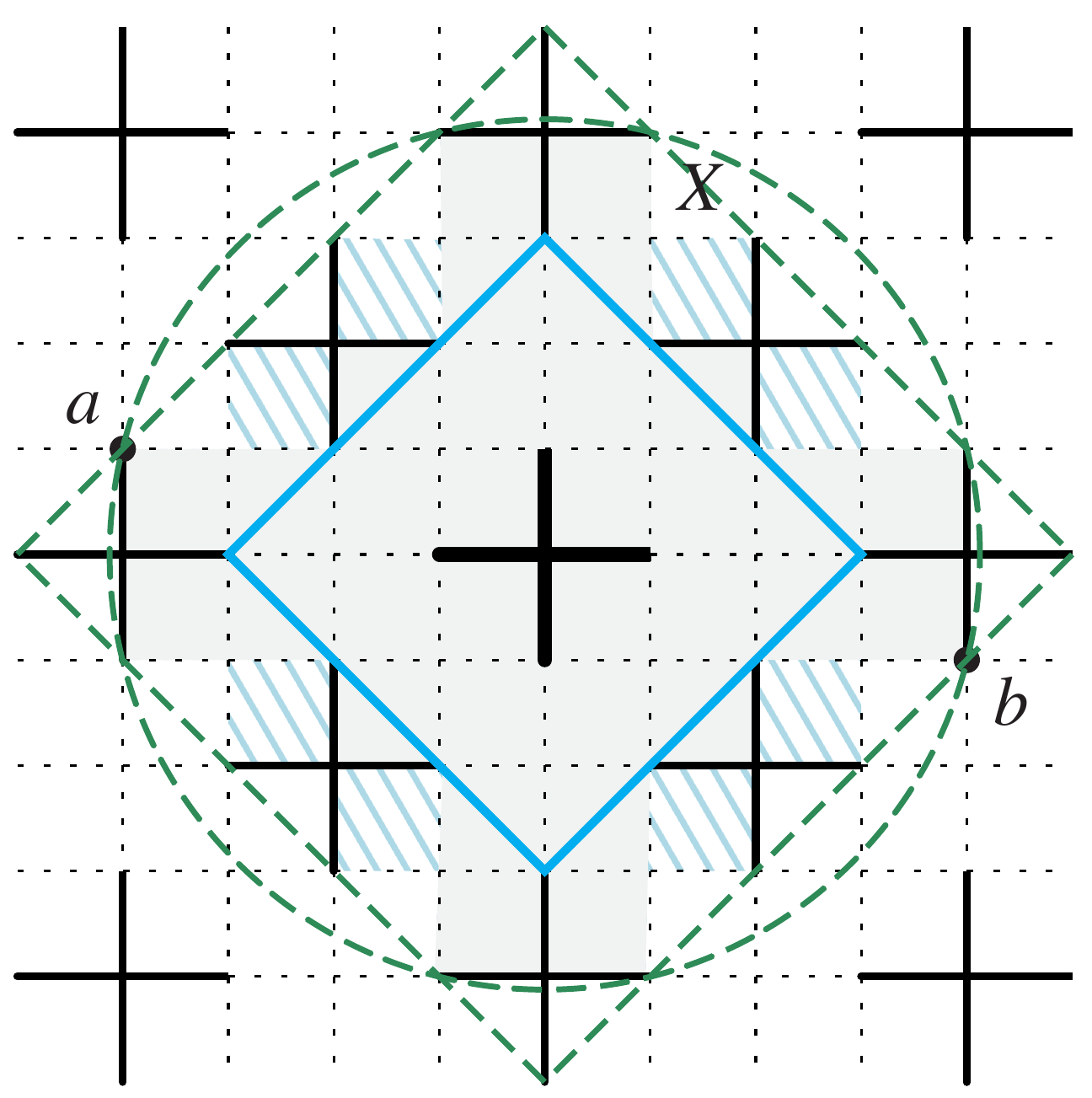}}
\hspace*{\fill}
{\includegraphics[width=.3\textwidth]{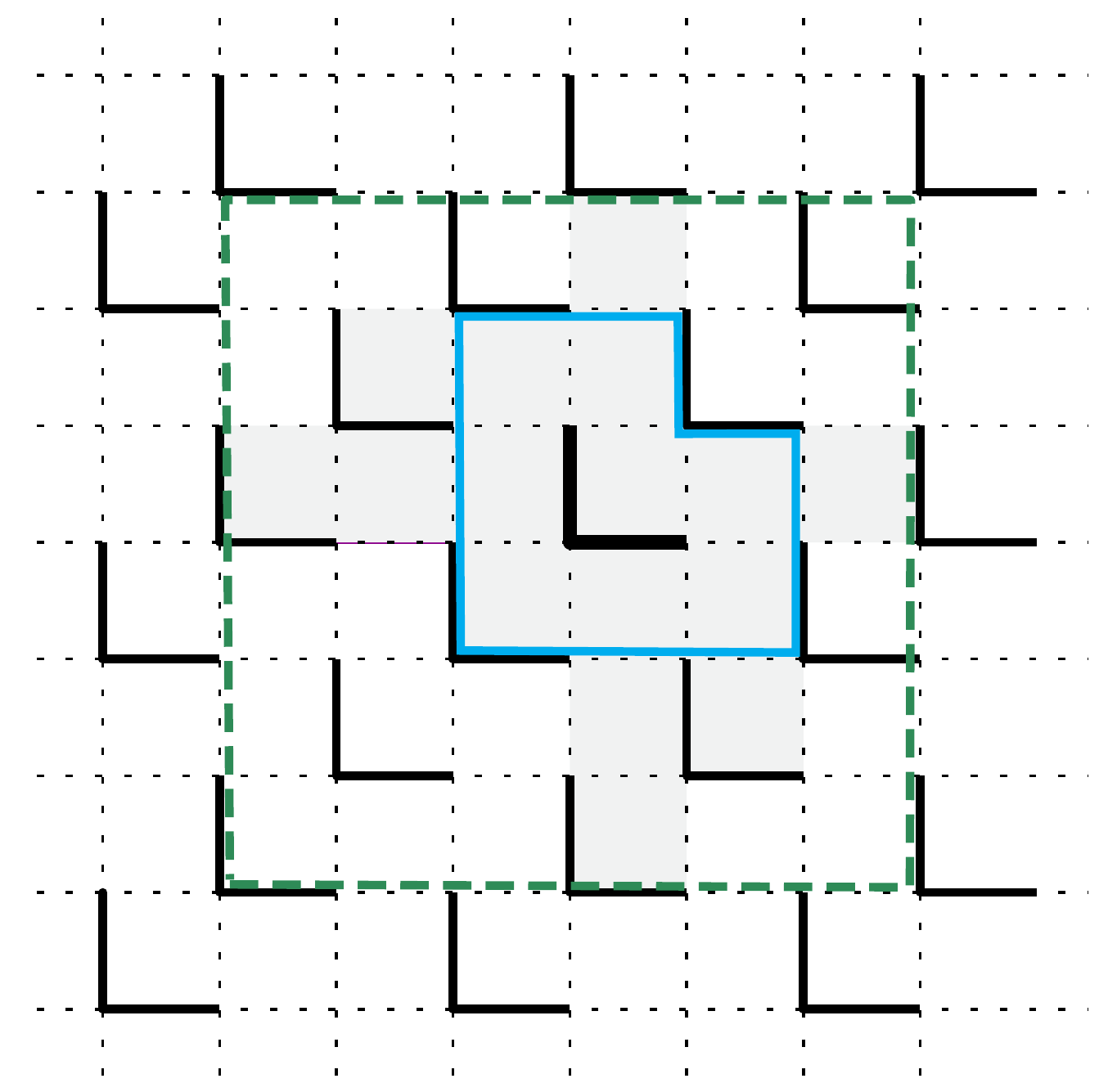}}
\hspace*{\fill}
\caption{A blocker shape $B$ (heavy   black) and  $C_\ell(B)$, the $L_\ell$-ball of diameter $f_\ell \delta$ centred at $B$'s anchor (dashed green) which is a diamond for $L_1$ (left), a circle for $L_2$ (left), and a square for $L_\infty$ (right).  The long cavities that touch $B$ are shaded gray. If a small rectangle $r$ owns $B$, then $p(r)$ lies 
in $C'$ (in cyan), and $r$ 
is contained in the union of the gray and blue-hatched regions. 
}
\label{fig:blocker-ball}
\end{figure}

We \changed{need a few claims} localizing $p^*(R)$ relative to $b(R)$:

\begin{claim}
\label{claim:rounded-point}
For any big rectangle $R$, the points $b(R)$ and $p^*(R)$ lie 
in one long cavity.
\end{claim}
\begin{proof}
Let $T$ be the rectangle with corners $p^*(R)$ and $b(R)$.  By definition of $b(R)$, there are no points of blocker-shapes in or on the boundary of $T$ except $b(R)$. 
Thus $T$ is contained in a cavity.  
Furthermore, if $T$ is contained in a square cavity, then we claim that $T$ does not contain the central grid point of the square cavity in its interior (otherwise 
$b(R)$ could not be 
the unique point of a blocker shape in $T$,
see Figure~\ref{fig:blocker-shapes}).  Thus $T$ is contained in a long cavity.
\end{proof}

\begin{claim}
\label{lem:ball-big}
Let $B$ be a blocker shape.  Let $R$ be a big rectangle with $B(R)=B$.  Then  $p^*(R)$ is contained in 
the ball $C_\ell(B)$. 
\end{claim}
\begin{proof}
By  the previous lemma, $p^*(R)$ lies in a long cavity  that contains a point of $B$.  
From  Figure~\ref{fig:blocker-ball} we see that any  long cavity that contains a point of $B$ lies inside the closed ball $C_\ell(B)$.  
Furthermore,  note that if $p^*(R)$ lies on the boundary of the ball, then it lies on a different blocker shape, contrary   to $B(R) =B$.
Thus $p^*(R)$ is contained in 
 $R_\ell(B)$.
\end{proof}

\begin{claim}
\label{lem:ball-small}
Let $B$ be a blocker shape.  Let $r$ be a small rectangle that owns $B$.  Then  $p^*(r)$ is contained in 
the ball $C_\ell(B)$. 
\end{claim}
\begin{proof}
Recall that $p(r)$ is the centre of $r$, and that, by the definition of ownership, for $\ell = 1, \infty$ we have $d_\ell(p(r),B) < \delta$, and for $\ell = 2$ we have $d_1(p(r),B) < \sqrt{2} \delta$.
Such points $p(r)$ lie in the open region $C'$ drawn in cyan in Figure~\ref{fig:blocker-ball}.
Here $C'$ is  
the Minkowski sum of an open ball with $B$ where we use an $L_1$-ball 
of radius $\delta$ for $\ell = 1$, an $L_1$-ball of radius $\sqrt{2}\delta$ for $\ell=2$ and an $L_\infty$-ball of radius $\delta$ for $\ell=\infty$.

We next show that $r$ is contained in the ball $C_\ell(B)$.
For $\ell=\infty$, $r$ must lie in a long cavity intersecting $C'$, i.e. in the open shaded gray region, thus in $C_\infty(B)$.

For $\ell=1,2$, $r$ is contained in a cavity that intersects $C'$.  
The union of the cavities that intersect $C'$ consists of the grey and blue-hatched region plus the grid square $X$ and its symmetric counterparts.  
But observe that a small rectangle that contains points of $X$ has a centre outside $C'$. 
Therefore $r$ is contained in $C_\ell(B)$.
\end{proof}

\begin{proof}[Proof of Lemma~\ref{lem:matchH}]
Let $R$ be a big rectangle and let $B = B(R)$.  By Claim~\ref{lem:ball-big}, $p^*(R)$ lies in $C_\ell(B)$.
If there is another big rectangle $R'$ with $B(R') = B$, then $p^*(R')$ also lies in $C_\ell(B)$.
If there is a small rectangle $r$ that owns $B$, then by Claim~\ref{lem:ball-small},  $p^*(r)$ lies in $C_\ell(B)$.

In either case, the distance between $p^*(R)$ and the other $p^*$ point is less than $f_\ell \delta$, since that is the diameter of $C_\ell(B)$.  Thus $\delta^* < f_\ell \delta$. 
\end{proof}

\section{Approximating the   optimization problem}
\label{sec:optimization-alg}

In this section we 
use 
{\sc Placement} to 
design approximation algorithms for the optimization version of the distant representatives problem for rectangles:

\begin{theorem}
\label{theorem:opt-alg}
There is an $f_\ell$-approximation algorithm for the distant representatives problem on rectangles in the $L_\ell$-norm, $\ell=1,2,\infty$ with run time $O(n^2 (\log n)^2)$ for $L_\infty$ and run time $O(n^2 {\rm \ polylog} (nD))$ for $L_1$ and $L_2$. 
Here (as before) $f_1 = 5, f_2 = \sqrt{34} \approx 5.83, f_\infty = 6$.
\end{theorem}

One complicating factor is that {\sc Placement} is not monotone, i.e., it may happen that {\sc Placement} fails for a value $\delta$ but succeeds for a larger value $\delta'$. We note that Cabello's algorithm~\cite{cabello2007approximation} behaves the same way.
We deal with $L_\infty$ in 
Section~\ref{sec:L-infty-optimization-alg} and with $L_1$ and $L_2$ in Section~\ref{sec:L-1,2-optimization-alg}.

We need some upper and lower bounds on $\delta^*$.  Note that if the input contains two rectangles that are single identical points, then
$\delta^*=0$.  Since this is easily detected, we assume from now on that no two input rectangles are single identical points.

\begin{claim}
\label{claim:delta-bounds}
We have $1/n \le \delta^* \le 2D$. 
\end{claim}
\begin{proof}
The upper bound is obvious.  For the lower bound, place
a grid of points 
distance $\frac{1}{n}$ 
 apart. All rectangles with non-zero dimensions will intersect at least 
$n+1$
points, and 
single-point
rectangles will hit one point.  
Since no two single-point rectangles are identical, they
can be matched to the sole grid point that overlaps the rectangle. The remaining rectangles can be matched trivially. 
\end{proof}
Note that a solution with distance at least $\tfrac{1}{n}$ can be found easily, following the steps above.

\subsection{Optimization problem for $L_\infty$}
\label{sec:L-infty-optimization-alg}

For the $L_\infty$ norm we use the following result about the possible optimum values; a proof is in
the full version. 

\begin{lemma}
\label{lemma:n-cubed}
In $L_\infty$, $\delta^*_\infty$ takes on one of the  $O(n^3)$ values of the form $(t-b)/k$ where $k \in \{1,\dots,n\}$ and $t,b$ are rectangle coordinates. 
\end{lemma}

\newcommand{\calD}{\ensuremath{\mathcal{D}}}

Let $\Delta$ be the set of $O(n^3)$ values from the lemma.
We can compute the set $\Delta$ in $O(n^3)$ time and 
sort it in $O(n^3 \log n)$ time.  
Say $\Delta = \{d_1,d_2,\dots,d_N\}$ in sorted order, and
set $c_i:=d_i/f_\infty$. 
Because of non-monotonicity, we cannot efficiently find the maximum $c_i$ for which {\sc Placement} succeeds.  
Instead, we use binary search to find $i$ such that {\sc Placement}($c_i$) succeeds but {\sc Placement}($c_{i+1}$) fails.
Therefore
$\delta^*_\infty < f_{\infty}c_{i+1} =  d_{i+1}$ which implies that  $\delta^*_\infty \le d_{i} = f_{\infty}c_i$ 
and the representative points found by {\sc Placement}($c_i$) provide an $f_\infty$-approximation.

To initialize the binary search, we first run {\sc Placement}$(c_N)$, and, if it succeeds, return its computed representative points
 since they provide an $f_\infty$-approximation
of the optimum assignment. 
Also note that {\sc Placement}$(c_1)$ must succeed---if it fails then 
$\delta_\infty^*<f_\infty c_1=d_1$,
which contradicts $\delta_\infty^*\in \Delta$.
Thus we begin with the interval $[1,N]$ and perform binary search on within this interval to find a value $i$ such that {\sc Placement}($c_i$) succeeds but {\sc Placement}($c_{i+1}$) fails.

We can get away with sorting just the $O(n^2)$ numerators and performing an implicit binary search, to avoid the cost of generating and sorting all of $\Delta$. Let $t$ be the sorted array of $O(n^2)$ numerators, which takes $O(n^2\log n)$ time to generate and sort. The denominators are just $[n]$, so there is no need to generate and sort it explicitly. Define the implicit \emph{sorted} matrix $a$, where $a[r, c] = t[r] / (n-c)$, for $0\le r\le |t|-1, 0 \le c \le n-1$. Each entry of $a$ can be computed in $O(1)$ time. Since the matrix is sorted, the matrix selection algorithm of Frederickson and Johnson \cite{frederickson1984generalized} can be used to get an element of $\Delta$ at the requested index in $O(n \log n)$ time. Using selection, one can perform the binary search on $\Delta$ implicitly. While accessing elements of $\Delta$ takes more time, it is still less than the time to call {\sc Placement} on the element accessed. Each iteration of the binary search is dominated by the runtime of {\sc Placement}, so the total runtime is 
$O(n^2(\log n)^2)$. 
This proves Theorem~\ref{theorem:opt-alg} for $L_\infty$.



\subsection{Optimization problem for $L_1$ and $L_2$}
\label{sec:L-1,2-optimization-alg}

In this section we give an approximation algorithm for the optimization version of distant representatives for rectangles in the $L_1$ and $L_2$ norms. 
Define a \emph{critical} value to be a right endpoint of an interval where {\sc Placement} succeeds.  See Figure~\ref{fig:critical}.

We use the following results whose proofs can be found below.

\begin{lemma}
\label{lem:closed=on-right}
{\sc Placement} succeeds at critical values, i.e., the intervals where {\sc Placement} succeeds are closed at the right.  Furthermore, a critical value provides an $f_\ell$-approximation.
\end{lemma}

\begin{figure}[htb]
    \centering
    \includegraphics[width=.58\textwidth]{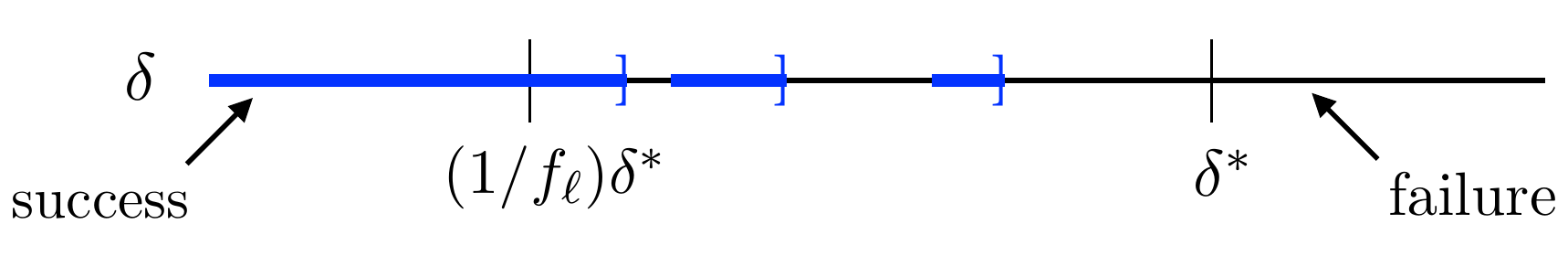}
    \caption{An illustration of critical values.}
    \label{fig:critical}
\end{figure}

Thus our problem reduces to finding a critical value.

\begin{lemma}
\label{lemma:detect-critical}
The {\sc Placement} algorithm can be modified to detect critical values.
\end{lemma}

\begin{lemma}
\label{lemma:critical-bits-bound}
In $L_1$ any critical value $\delta$  is a rational number with numerator and denominator at most $4Dn$.  In $L_2$ for any critical value $\delta$, $\delta^2$ is a rational number with numerator and denominator at most $8D^2n^2$. 
\end{lemma}

Based on these lemmas, we use continued fractions to find a critical value.
We need the following properties of continued fractions.
\begin{enumerate}
\item A continued fraction has the form $a_0 + \frac{1}{a_1 + \frac{1}{a_2 + \cdots \frac{1}{a_k}}}$, where the $a_i$'s are natural numbers.

\item Every positive rational number $\frac{a}{b}$ has a continued fraction representation. Furthermore, 
the number of terms, $k$, is $O(\log (\max\{a,b\}))$.  This follows from the fact that computing the continued fraction representation of $\frac{a}{b}$ exactly parallels the Euclidean algorithm; see~\cite[Theorem 4.5.2]{Bach-Shallit} or the wikipedia page on the Euclidean algorithm~\cite{wiki_euclididean}.
For the same reason, 
each $a_i$ is bounded by $\max\{a,b\}$.

\item Suppose we don't know $\frac{a}{b}$ explicitly, but we know some bound $G$ such that $a,b \le G$, 
and we have a test of whether a partial continued fraction is 
greater than, less than, or equal to $\frac{a}{b}$.  Then we can find the continued fraction representation of $\frac{a}{b}$ as follows. For $i=1 \ldots \log G$, use binary search on $[2..G]$ to find the best value for $a_i$.  
Note that the continued fraction with $i$ terms is increasing in $a_i$ for even $i$ and decreasing for odd $i$, and we adjust the binary search correspondingly.
In each step, we have a lower bound and an upper bound on $\frac{a}{b}$ and the step shrinks the interval.
If the test runs in time $T$, then
the time to find the continued fraction for $\frac{a}{b}$ is $O(T (\log G)^2)$, plus the cost of doing arithmetic on continued fractions (with no $T$ factor).
\end{enumerate}

\subparagraph*{Algorithm for $L_1, L_2$.}
Run the continued fraction algorithm using {\sc Placement} (enhanced to detect a critical value) as the test. The only difference from the above description is that we do not have a specific target $\frac{a}{b}$; rather, our interval contains at least one critical value and we search until we find one.  At any point we have two values $b_l$ and $b_u$ both represented as continued fractions, where {\sc Placement} succeeds at $b_l$ and fails at $b_u$, so there is at least one critical value between them.
We can use the initial interval $[1/n, 2D]$.  
To justify this, note that if 
{\sc Placement}($\frac{1}{n}$) fails , then $f_\ell/n > \delta^*$ by Theorem \ref{thm:placement}, so we get an 
$f_\ell$-approximation by using the representative points for $\delta = \frac{1}{n}$
(see the remark after Claim~\ref{claim:delta-bounds}).

For the runtime, we use the bound $O(T(\log G)^2)$ from point 3 above, plugging in $T = O(n^2 \log n)$ for {\sc Placement} and the bounds on $G$ from Lemma~\ref{lemma:critical-bits-bound}, to obtain a runtime of $O(n^2 {\rm polylog}(nD))$, 
which proves Theorem~\ref{theorem:opt-alg} for $L_1$ and $L_2$.

The run-time for Theorem~\ref{theorem:opt-alg} can actually be improved to $O(n^2(\log n)^2)$ (i.e., without the dependence on $\log D$) with an approach that is very specific to the problem at hand (and similar to Cabello's approach).  
The details are complicated for such a relatively small improvement and hence omitted here.

\subparagraph*{Missing proofs.}  For space reasons we can here only give the briefest sketch of the proofs of Lemmas~\ref{lem:closed=on-right}, \ref{lemma:detect-critical}, and~\ref{lemma:critical-bits-bound}; details are in 
the full version.
A crucial ingredient is to study what must have happened if {\sc Placement}($\delta$) goes from success to failure (when viewing its outcome as a function that changes over time $\delta$).

\begin{observation}
\label{obs:events}
Assume {\sc Placement}($\delta$) succeeds but {\sc Placement}($\delta'$) fails for some $\delta'>\delta$.  Then at least one of the following
events occurs 
as we go from $\delta$ to $\delta'$:
\begin{enumerate}
\item the set of small/big rectangles changes,
\item the distance between the centres of two small rectangles equals $\hat{\delta}$ for some $\delta\leq \hat{\delta}<\delta'$,
\item the set of blocker-shapes owned by a small rectangle increases,
\item the set of blocker-shapes intersecting a big rectangle decreases.
\end{enumerate}
\end{observation}

Roughly speaking,  Lemma~\ref{lem:closed=on-right} 
can now be  shown by arguing that such events do not happen in a sufficiently small time-interval before {\sc Placement} fails (hence the intervals where it fails are open on the left).  Lemma~\ref{lemma:critical-bits-bound} holds because there necessarily must have been an event at time $\delta$, and we can analyze the coordinates when events happen.  
Finally Lemma~\ref{lemma:detect-critical}  is achieved by running {\sc Placement} at time $\delta$ and also symbolically at time $\delta+\varepsilon$.



\section{Hardness results}
\label{sec:hardness}

In this section we outline NP-hardness and 
APX-hardness 
results for the 
distant representatives problem.
For complete details see the full version of the paper.
We first show that, even for the special case of unit horizontal segments, the decision version of the problem is NP-complete
for $L_1$ and $L_\infty$ and NP-hard for $L_2$ (where bit  complexity issues prevent us from placing the problem in NP). This $L_2$ result was proved previously by Roeloffzen in  his Master's thesis~\cite[Section 2.3]{roeloffzenfinding} but 
we add details regarding bit complexity that were missing from his proof.

Next, we enhance our reductions to ``gap-producing reductions''  to obtain lower bounds on the approximation  factors that can be  achieved in polynomial time.  Since our goal is to compare with 
 our approximation algorithms
 for 
 rectangles, 
 we consider the more general case of horizontal and vertical segments in  the plane (not just unit horizontals). 
Our main result is that,
assuming P $\ne$ NP, no polynomial time approximation algorithm achieves a factor better than
$1.5$ in $L_1$ and $L_\infty$ and $1.4425$ in $L_2$.

Our reductions are from the NP-complete problem Monotone Rectilinear Planar 3-SAT~\cite{deberg2010optimal} in which
each clause has either three positive literals or three negative literals, each variable is represented by a thin vertical rectangle at $x$-coordinate $0$,
each positive [negative] clause is represented by a thin vertical rectangle at a positive [negative, resp.] $x$-coordinate, and there is a horizontal line segment joining any variable to any  clause that contains it.  See Figure~\ref{fig:3-SAT-plan}(a) for an example instance of the problem.  
For 
$n$ variables and $m$ clauses, the representation can be on an $O(m) \times O(n+m)$ grid. 

\subsection{ NP-hardness }
\label{sec:NP-hardness}

\begin{theorem}
The decision version of the distant representatives problem for unit horizontal segments in the $L_1, L_2$ or $L_{\infty}$ norm is NP-hard.  
\end{theorem}

Lemma~\ref{lemma:n-cubed} implies that the decision problem lies in NP for the $L_\infty$ norm, even for  rectangles.  
In the full version 
we show the same for $L_1$, and  we discuss the bit complexity issues that prevent us from placing the decision problem in NP  for the $L_2$  norm.

For our reduction from Monotone Rectilinear Planar 3-SAT we first modify the representation so that 
each clause rectangle has fixed height and is connected to its three literals via three ``wires''---the middle one remains horizontal, the bottom one bends to enter the clause rectangle from the bottom, and the top one bends twice to enter the clause rectangle from the far side.  See Figure~\ref{fig:3-SAT-plan}(b).
Each wire is directed from the variable to the clause, and represents a literal. 
The representation is still on an 
$O(m) \times O(n+m)$ grid.

To complete the  reduction to the distant representatives problem we 
replace the rectangles with variable and clause gadgets constructed from unit horizontal intervals, and also implement the  wires using such intervals.  
The details, which can be found in the full version of the paper, 
depend on the norm $L_\ell$, $\ell  = 1, 2, \infty$.
We also set a value of $\delta_\ell$ to obtain a decision problem that asks if  there is an assignment of a representative point to each interval that is \emph{valid}, i.e., such that no two points are closer than $\delta_\ell$.  
We set $\delta_1 = 2$, $\delta_2 = \frac{13}{5}$, and $\delta_\infty = \frac{1}{2}$.
An example of the construction for $L_1$ (with $\delta_1 = 2$) is shown in Figure~\ref{fig:3-SAT-plan}(c). 

\begin{figure}[tb]
    \centering
    \includegraphics[width=.87\textwidth]{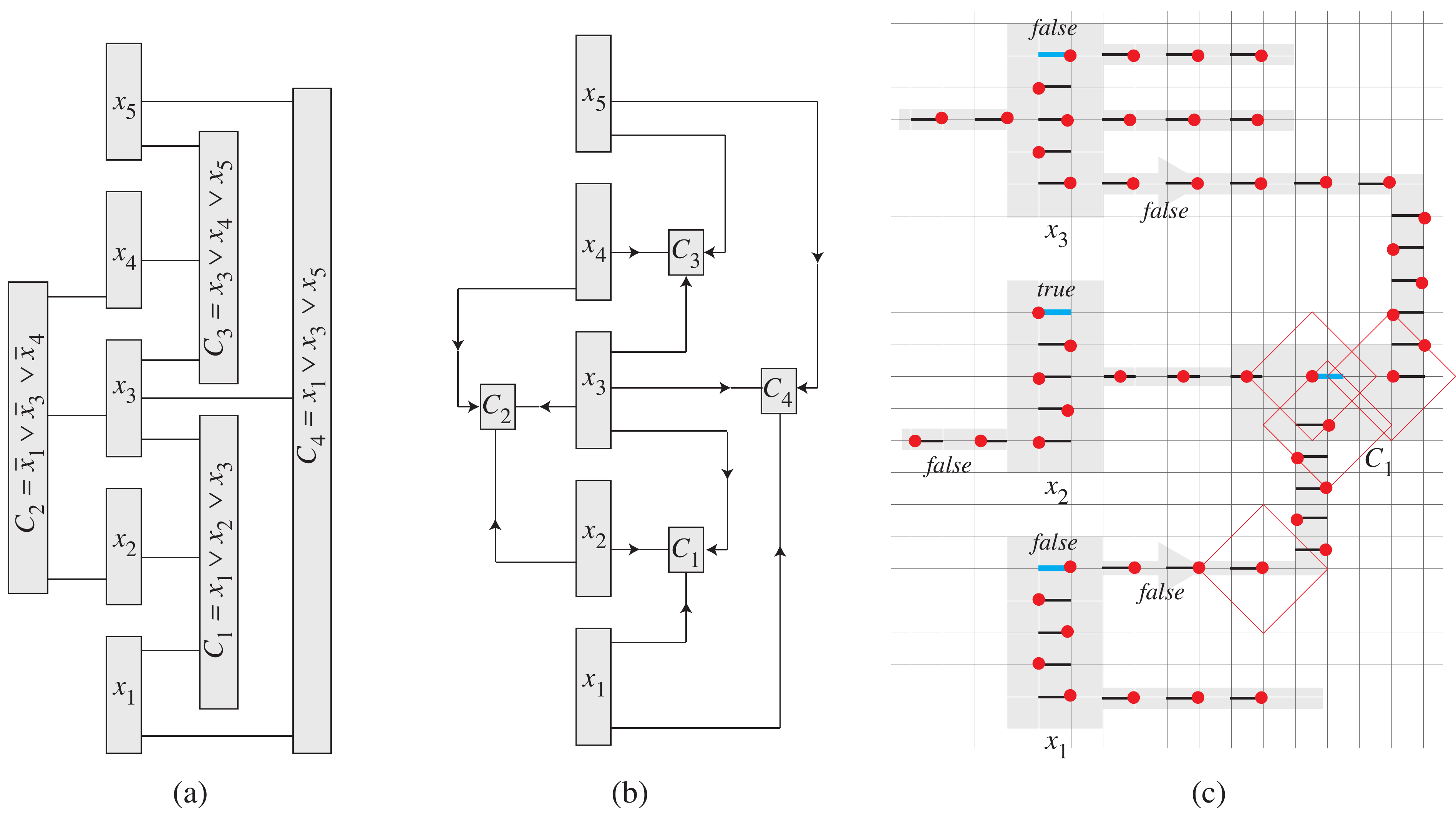}
    \caption{(a)  An instance of Monotone Rectilinear Planar 3-SAT. (b) The modified representation used for our NP-hardness proofs, with wires from variable to clause gadgets. 
    (c) A detail of our NP-hardness construction for clause $C_1 = x_1 \vee x_2 \vee x_3$ in the $L_1$ norm
    showing how the truth-value setting $x_1 = $ False, $x_2 = $ True,  $x_3 = $ False, permits representative points (shown as red dots) at distance at least $\delta_1 = 2$.
}
    \label{fig:3-SAT-plan}
\end{figure}

For the $L_2$ norm, the bit complexity issue missed in Roeloffzen's reduction~\cite[Section 2.3]{roeloffzenfinding} is that the interval endpoints and their distances must have polynomially-bounded bit complexity.
We resolve this by using 
Pythagorean triples (see Figure~\ref{fig:approx-plan}(a)).

\subsection{APX-hardness}
\label{sec:APXX-hardness}

In this section, we prove hardness-of-approximation results for the distant representatives problem on  horizontal and vertical segments in the plane.
Specifically, 
we prove lower bounds on the approximation factors that can be achieved in polynomial time, assuming P $\ne$ NP.

\begin{theorem} For $\ell = 1, 2, \infty$, let $g_\ell$ be the constant shown in Table~\ref{Tab:ratios}.  Suppose P $\ne NP$.  Then, for the $L_\ell$ norm,  there is no polynomial time algorithm with approximation factor less than $g_\ell$ for the distant representatives problem for horizontal and vertical segments.
\label{PTAS_claim}
\end{theorem}
\begin{table}[ht]
    \centering
    
    \begin{tabular}{c|c c c}
           & $L_1$ & $L_2$ & $L_\infty$\\
           \hline
           lower bound &  $g_1 = 1.5$ & $g_2 = 1.4425$ 
           & $g_{\infty} = 1.5$ \\
    \end{tabular}
    
    \caption{Best approximation ratios that can be achieved unless P=NP.} 
    \label{Tab:ratios}
\end{table}

We prove Theorem~\ref{PTAS_claim} using a 
\emph{gap reduction}.  
This standard approach is based on the fact that if there were polynomial time approximation algorithms with approximation factors better than $g_\ell$ then the \emph{gap versions} of the problem (as stated below) would be solvable in polynomial time.
Thus, proving that the gap versions are NP-hard implies that there are no polynomial time $g_\ell$-approximation algorithms unless P=NP.

Recall that $\delta^*_\ell$ is the max over all assignments of representative points, of the min distance between two points.

\begin{quote}
{\bf Gap Distant Representatives Problem.}\\
{\bf Input:} A set $I$ of horizontal and vertical segments in the plane.\\
{\bf Output:} 
\begin{itemize}
    \item YES if $\delta^*_\ell(I) \ge 1$;
    \item NO if $\delta^*_\ell(I) \le 1/g_\ell$;
    \item  and it does not matter what the output is for other inputs.
\end{itemize}

\end{quote}

To prove Theorem~\ref{PTAS_claim} it therefore suffices to prove:

\begin{theorem}
The Gap Distant Representatives problem is NP-hard.
\label{theorem:gap-NP-hard}
\end{theorem}

This is proved via a reduction from Monotone Rectilinear Planar 3-SAT, much like in the previous section. 
The gadgets are simpler because we can use vertical segments, but we must prove stronger properties. 
Given an instance $\Phi$ of Monotone Rectilinear Planar 3-SAT
we construct in polynomial time a set of horizontal and vertical segments $I$ 
such that:

\begin{claim}
\label{claim:if-SAT}
If $\Phi$ is satisfiable then $\delta^*_\ell(I) = 1$.
\end{claim}

\begin{claim}
\label{claim:if-not-SAT}
If $\Phi$ is not satisfiable then $\delta^*_\ell(I) \le 1/g_\ell$.
\end{claim}

Thus a polynomial time algorithm for the Gap Distant Representatives problem yields a polynomial time algorithm for Monotone Rectilinear Planar 3-SAT.
We give some of the reduction details, but defer the proofs of the claims to the full version.

\paragraph*{Reduction details}

We reduce directly from Monotone Rectilinear Planar 3-SAT.  

\begin{figure}[htb]
    \centering
    \includegraphics[width=.8\textwidth]{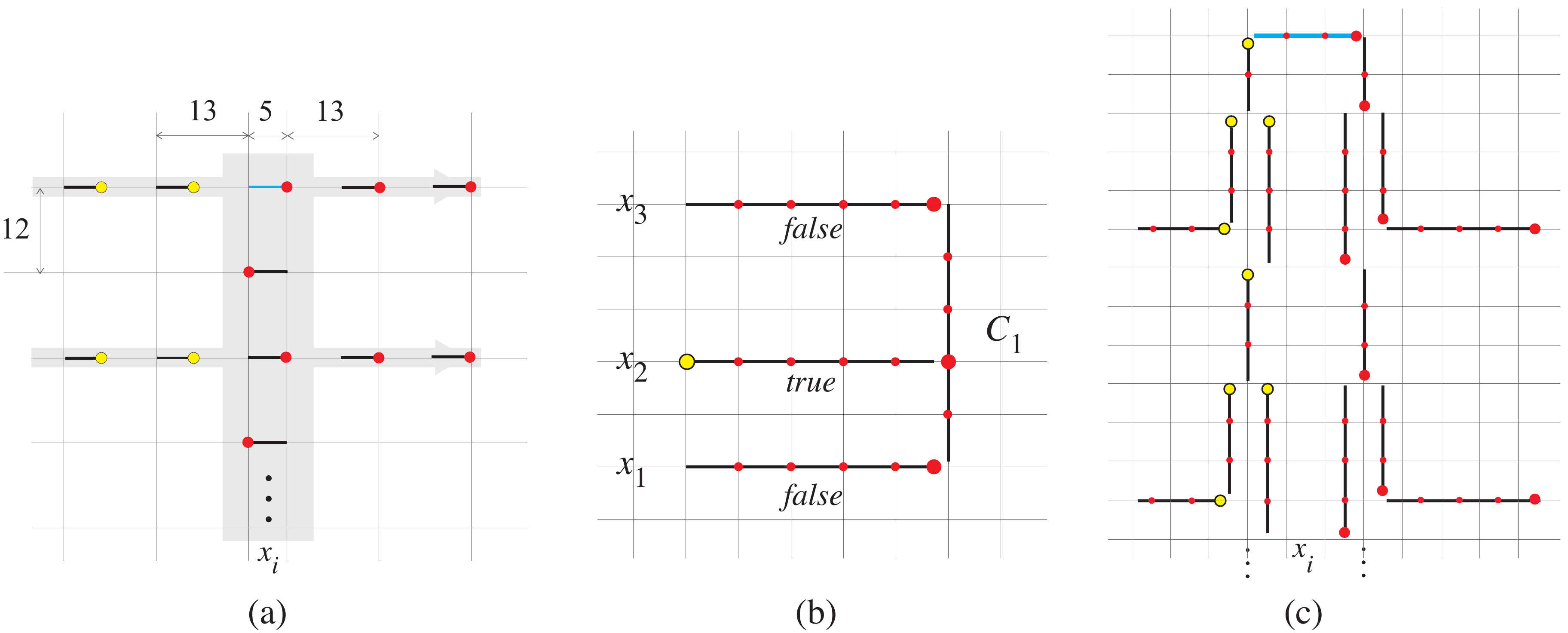}
    \caption{
    (a) A variable gadget for NP-hardness for $L_2$, based on Pythagorean triple $5,12, 13$.  To achieve $\delta = 13$ the representative point for the variable interval (in cyan) is forced to the left (true) or the right (false) in which case intervals on the right are also forced.
    (b) A clause gadget for the APX-hardness reduction, with three horizontal wires attached.  For clarity, segments are not drawn all the way to their endpoints.
     Wires $x_1$ and $x_2$ are in the false setting and wire $x_2$ is in the true setting, which allows the representative point for $C_1$ to be placed where the $x_2$ wire meets it, while keeping representative points at least distance 1 apart.
    (c) The basic splitter gadget for APX-hardness for $L_{\infty}$ placed on the half grid and showing two wires extending left and two right.  The variable segment (in thick cyan) for the variable $x_i$ has its representative point (the large red dot) at the right, which is the false setting.
    The representative points shown by large red/yellow dots are distance at least 1 apart in $L_\infty$.
     }
    \label{fig:approx-plan}
\end{figure}

\subparagraph*{Wire.} 
A wire is a long horizontal segment 
with 0-length segments at unit distances along it, except at its left and right endpoints. 
See Figure~\ref{fig:approx-plan}(b). 
The representative point for a 0-length segment must be the single point in the segment (shown as small red dots in the figure).
As before, a wire is directed from the variable gadget to the clause gadget.
We distinguish a 
``false setting'' where the wire has its representative point within distance 1 of its forward end (at the clause gadget) and a ``true setting'' where the wire has its representative point within distance 1 of its tail end (at the variable gadget).

\subparagraph*{Clause gadget.}  
A clause gadget is a vertical segment.  Three wires corresponding to the three literals in the clause meet the vertical segment as shown in Figure~\ref{fig:approx-plan}(b).  There are 0-length segments at unit distances along the clause interval except where the three wires meet it.

\subparagraph*{Variable gadget.} 
A variable segment has length 3, with two 0-length segments placed 1 and 2 units from the endpoints.  A representative point in the right half corresponds to a false value for the variable, and a representative point in the left half corresponds to a true value. 
In order to transmit the variable's value to all the connecting horizontal wires we build a ``splitter'' gadget.  The basic splitter gadget for $L_\infty$ is shown in Figure~\ref{fig:approx-plan}(c).  The same splitter gadget works for the other norms but we can improve the lower bounds using modified splitter gadgets as described 
in the full version.


\section{Conclusions}

We gave good approximation algorithms for the distant representatives problem for rectangles in the plane using a new technique of ``imprecise discretization'' where we limit the choice of representative points not to a discrete set but to a set of one-dimensional ``shapes''.  This technique may be more widely applicable, and can easily be tailored, for example by using a weighted matching algorithm to prefer representative points near the centres of rectangles. 

We also gave the first explicit lower bounds on approximation factors that can be achieved in polynomial time for distant representatives problems.  

Besides the obvious questions of improving the approximation factors, the run-times, or the lower bounds, we mention several other avenues for further research.

\begin{enumerate}
\item Is the distant representatives problem for rectangles in $L_2$ hard for existential theory of  the reals? Recently, some packing problems have been proved $\exists \mathbb{R}$-complete~\cite{abrahamsen2020framework}, but they seem substantially harder.

\item Is there a good [approximation] algorithm for any version of distant representatives for a \emph{lexicographic} objective function.  For example, suppose we wish to maximize the smallest distance between points, and, subject to that, maximize the second smallest distance, and so on.
Or suppose we ask to  lexicographically maximize the sorted vector consisting of the $n$ distances from each chosen point to its nearest neighbour.
For the case of \emph{ordered} line segments in 1D there is a linear time algorithm to lexicographically minimize the sorted vector of distances between successive pairs of points~\cite{biedl2021dispersion}.  It is an open problem to extend this to unordered line segments.

\item What about weighted versions of distant representatives? Here each rectangle $R$ has a weight $w(R)$, and rather than packing disjoint balls of radius $\delta$ we pack disjoint balls of radius $w(R) \delta$ centred at a representative point $p(R)$ in $R$.  Again, there is a solution for ordered line segments in 1D~\cite{biedl2021dispersion}. 

\end{enumerate}

\paragraph*{Acknowledgements.}
We thank Jeffrey Shallit for help with continued fractions, and we thank anonymous reviewers for their helpful comments.

\bibliographystyle{plainurl}   
\bibliography{distant}

\begin{thebibliography}{10}

\bibitem{abrahamsen2020framework}
Mikkel Abrahamsen, Tillmann Miltzow, and Nadja Seiferth.
\newblock Framework for $\exists \mathbb{R}$-completeness of two-dimensional
  packing problems.
\newblock {\em arXiv preprint arXiv:2004.07558}, 2020.
\newblock URL: \url{https://arXiv.org/abs/2004.07558}.

\bibitem{Bach-Shallit}
Eric Bach and Jeffrey Shallit.
\newblock {\em Algorithmic Number Theory: Efficient Algorithms}, volume~1.
\newblock MIT press, 1996.

\bibitem{baur2001approximation}
Christoph Baur and S{\'a}ndor~P. Fekete.
\newblock Approximation of geometric dispersion problems.
\newblock {\em Algorithmica}, 30(3):451--470, 2001.
\newblock URL: \url{https://doi.org/10.1007/s00453-001-0022-x}.

\bibitem{biedl2021dispersion}
Therese Biedl, Anna Lubiw, Anurag Murty~Naredla, Peter Dominik~Ralbovsky, and
  Graeme Stroud.
\newblock Dispersion for intervals: A geometric approach.
\newblock In {\em Symposium on Simplicity in Algorithms (SOSA)}, pages 37--44.
  SIAM, 2021.
\newblock URL: \url{https://doi.org/10.1137/1.9781611976496.4}.

\bibitem{cabello2007approximation}
Sergio Cabello.
\newblock Approximation algorithms for spreading points.
\newblock {\em Journal of Algorithms}, 62(2):49--73, 2007.
\newblock URL: \url{https://doi.org/10.1016/j.jalgor.2004.06.009}.

\bibitem{cardinal2015computational}
Jean Cardinal.
\newblock Computational geometry column 62.
\newblock {\em ACM SIGACT News}, 46(4):69--78, 2015.
\newblock URL: \url{https://doi.org/10.1145/2852040.2852053}.

\bibitem{chambers2017connectivity}
Erin Chambers, Alejandro Erickson, S{\'a}ndor~P. Fekete, Jonathan Lenchner,
  Jeff Sember, Venkatesh Srinivasan, Ulrike Stege, Svetlana Stolpner,
  Christophe Weibel, and Sue Whitesides.
\newblock Connectivity graphs of uncertainty regions.
\newblock {\em Algorithmica}, 78(3):990--1019, 2017.
\newblock URL: \url{https://doi.org/10.1007/s00453-016-0191-2}.

\bibitem{chan2003polynomial}
Timothy~M. Chan.
\newblock Polynomial-time approximation schemes for packing and piercing fat
  objects.
\newblock {\em Journal of Algorithms}, 46(2):178--189, 2003.

\bibitem{chen2019efficient}
Jing Chen, Bo~Li, and Yingkai Li.
\newblock Efficient approximations for the online dispersion problem.
\newblock {\em SIAM Journal on Computing}, 48(2):373--416, 2019.
\newblock URL: \url{https://doi.org/10.1137/17M1131027}.

\bibitem{deberg2010optimal}
Mark de~Berg and Amirali Khosravi.
\newblock Optimal binary space partitions in the plane.
\newblock In {\em International Computing and Combinatorics Conference}, pages
  216--225. Springer, 2010.
\newblock URL: \url{https://doi.org/10.1007/978-3-642-14031-0_25}.

\bibitem{demaine2010circle}
Erik~D. Demaine, S{\'a}ndor~P. Fekete, and Robert~J. Lang.
\newblock Circle packing for origami design is hard.
\newblock {\em arXiv preprint arXiv:1008.1224}, 2010.
\newblock URL: \url{https://arxiv.org/abs/1008.1224}.

\bibitem{doddi1997map}
Srinivas Doddi, Madhav~V. Marathe, Andy Mirzaian, Bernard~M.E. Moret, and
  Binhai Zhou.
\newblock Map labeling and its generalizations.
\newblock In {\em Proc. 8th Ann. ACM/SIAM Symp. Discrete Algs.(SODA97)}, pages
  148--157. SIAM, 1997.

\bibitem{dorrigiv2015minimum}
Reza Dorrigiv, Robert Fraser, Meng He, Shahin Kamali, Akitoshi Kawamura,
  Alejandro L{\'o}pez-Ortiz, and Diego Seco.
\newblock On minimum-and maximum-weight minimum spanning trees with
  neighborhoods.
\newblock {\em Theory of Computing Systems}, 56(1):220--250, 2015.
\newblock URL: \url{https://doi.org/10.1007/s00224-014-9591-3}.

\bibitem{dumitrescu2012dispersion}
Adrian Dumitrescu and Minghui Jiang.
\newblock Dispersion in disks.
\newblock {\em Theory of Computing Systems}, 51(2):125--142, 2012.
\newblock URL: \url{https://doi.org/10.1007/s00224-011-9331-x}.

\bibitem{dumitrescu2015systems}
Adrian Dumitrescu and Minghui Jiang.
\newblock Systems of distant representatives in {E}uclidean space.
\newblock {\em Journal of Combinatorial Theory, Series A}, 134:36--50, 2015.
\newblock URL: \url{https://doi.org/10.1016/j.jcta.2015.03.006}.

\bibitem{erickson2020smoothing}
Jeff Erickson, Ivor van~der Hoog, and Tillmann Miltzow.
\newblock Smoothing the gap between {NP} and $\exists \mathbb{R}$.
\newblock In {\em 2020 IEEE 61st Annual Symposium on Foundations of Computer
  Science (FOCS)}, pages 1022--1033. IEEE, 2020.

\bibitem{fiala2005systems}
Ji{\v{r}}{\'\i} Fiala, Jan Kratochv{\'\i}l, and Andrzej Proskurowski.
\newblock Systems of distant representatives.
\newblock {\em Discrete Applied Mathematics}, 145(2):306--316, 2005.
\newblock URL: \url{https://doi.org/10.1016/j.dam.2004.02.018}.

\bibitem{formann1991packing}
Michael Formann and Frank Wagner.
\newblock A packing problem with applications to lettering of maps.
\newblock In {\em Proceedings of the Seventh Annual Symposium on Computational
  Geometry}, pages 281--288, 1991.

\bibitem{frederickson1984generalized}
Greg~N. Frederickson and Donald~B. Johnson.
\newblock Generalized selection and ranking: sorted matrices.
\newblock {\em SIAM Journal on Computing}, 13(1):14--30, 1984.

\bibitem{garey1981scheduling}
Michael~R. Garey, David~S. Johnson, Barbara~B. Simons, and Robert~Endre Tarjan.
\newblock Scheduling unit--time tasks with arbitrary release times and
  deadlines.
\newblock {\em SIAM Journal on Computing}, 10(2):256--269, 1981.
\newblock URL: \url{https://doi.org/10.1137/0210018}.

\bibitem{hall1935representatives}
Philip Hall.
\newblock On representatives of subsets.
\newblock {\em Journal of the London Mathematical Society}, 1(1):26--30, 1935.

\bibitem{hifi2009literature}
Mhand Hifi and Rym M'hallah.
\newblock A literature review on circle and sphere packing problems: Models and
  methodologies.
\newblock {\em Advances in Operations Research}, 2009, 2009.
\newblock URL: \url{https://doi.org/10.1155/2009/150624}.

\bibitem{hochbaum1985approximation}
Dorit~S. Hochbaum and Wolfgang Maass.
\newblock Approximation schemes for covering and packing problems in image
  processing and {VLSI}.
\newblock {\em Journal of the ACM (JACM)}, 32(1):130--136, 1985.
\newblock URL: \url{https://doi.org/10.1016/j.orl.2010.07.004}.

\bibitem{Hopcroft73Karp}
John~E. Hopcroft and Richard~M. Karp.
\newblock An n\({}^{\mbox{5/2}}\) algorithm for maximum matchings in bipartite
  graphs.
\newblock {\em {SIAM} J. Comput.}, 2(4):225--231, 1973.
\newblock URL: \url{https://doi.org/10.1137/0202019}, \href
  {http://dx.doi.org/10.1137/0202019} {\path{doi:10.1137/0202019}}.

\bibitem{leung1990packing}
Joseph~Y.T. Leung, Tommy~W. Tam, Chin~S. Wong, Gilbert~H. Young, and
  Francis~Y.L. Chin.
\newblock Packing squares into a square.
\newblock {\em Journal of Parallel and Distributed Computing}, 10(3):271--275,
  1990.
\newblock URL: \url{https://doi.org/10.1016/0743-7315(90)90019-L}.

\bibitem{li2018dispersing}
Shimin Li and Haitao Wang.
\newblock Dispersing points on intervals.
\newblock {\em Discrete Applied Mathematics}, 239:106--118, 2018.
\newblock URL: \url{https://doi.org/10.1016/j.dam.2017.12.028}.

\bibitem{loffler2010-CH}
Maarten L{\"o}ffler and Marc van Kreveld.
\newblock Largest and smallest convex hulls for imprecise points.
\newblock {\em Algorithmica}, 56(2):235, 2010.
\newblock URL: \url{https://doi.org/10.1007/s00453-008-9174-2}.

\bibitem{loffler2010-box}
Maarten L{\"o}ffler and Marc van Kreveld.
\newblock Largest bounding box, smallest diameter, and related problems on
  imprecise points.
\newblock {\em Computational Geometry}, 43(4):419--433, 2010.
\newblock URL: \url{https://doi.org/10.1016/j.comgeo.2009.03.007}.

\bibitem{marshall2005scaled}
Roger~J. Marshall.
\newblock Scaled rectangle diagrams can be used to visualize clinical and
  epidemiological data.
\newblock {\em Journal of Clinical Epidemiology}, 58(10):974--981, 2005.
\newblock URL: \url{https://doi.org/10.1016/j.jclinepi.2005.01.018}.

\bibitem{datastruct}
Christian~Worm Mortensen.
\newblock Fully-dynamic two dimensional orthogonal range and line segment
  intersection reporting in logarithmic time.
\newblock In {\em Proceedings of the Fourteenth Annual ACM-SIAM Symposium on
  Discrete Algorithms}, SODA '03, page 618–627, USA, 2003. Society for
  Industrial and Applied Mathematics.

\bibitem{roeloffzenfinding}
M.J.M. Roeloffzen.
\newblock Finding structures on imprecise points.
\newblock Master's thesis, TU Eindhoven, 2009.
\newblock URL:
  \url{https://www.win.tue.nl/~mroeloff/papers/thesis-roeloffzen2009.pdf}.

\bibitem{schrijver2003combinatorial}
Alexander Schrijver.
\newblock {\em Combinatorial Optimization: Polyhedra and Efficiency},
  volume~24.
\newblock Springer Science \& Business Media, 2003.

\bibitem{simons1978fast}
Barbara Simons.
\newblock A fast algorithm for single processor scheduling.
\newblock In {\em 19th Annual Symposium on Foundations of Computer Science},
  pages 246--252. IEEE, 1978.
\newblock URL: \url{https://doi.org/10.1109/SFCS.1978.4}.

\bibitem{szabo2007new}
P{\'e}ter~G{\'a}bor Szab{\'o}, Mihaly~Csaba Mark{\'o}t, Tibor Csendes, Eckard
  Specht, Leocadio~G. Casado, and Inmaculada Garc{\'\i}a.
\newblock {\em New Approaches to Circle Packing in a Square: with Program
  Codes}, volume~6.
\newblock Springer Science \& Business Media, 2007.

\bibitem{wagner2001three}
Frank Wagner, Alexander Wolff, Vikas Kapoor, and Tycho Strijk.
\newblock Three rules suffice for good label placement.
\newblock {\em Algorithmica}, 30(2):334--349, 2001.

\bibitem{wiki_euclididean}
{Wikipedia contributors}.
\newblock Euclidean algorithm --- {Wikipedia}{,} the free encyclopedia, 2021.
\newblock [Online; accessed 28-June-2021].
\newblock URL:
  \url{https://en.wikipedia.org/w/index.php?title=Euclidean_algorithm&oldid=1027503317}.

\bibitem{wolff2001simple}
Alexander Wolff, Lars Knipping, Marc van Kreveld, Tycho Strijk, and Pankaj~K.
  Agarwal.
\newblock A simple and efficient algorithm for high-quality line labeling.
\newblock In Peter Atkinson, editor, {\em GIS and GeoComputation}. Taylor and
  Francis, 2000.
\newblock \href {http://dx.doi.org/https://doi.org/10.1201/9781482268263}
  {\path{doi:https://doi.org/10.1201/9781482268263}}.

\end{thebibliography}

\newpage

\begin{appendix}

\section{Decreasing the runtime for the decision problem}
\label{app:runtime}

In the main part of the paper, we proved a run-time of $O(n^{2.5})$ for {\sc Placement}, with the bottleneck being the runtime for finding the matching (all other aspects take $O(n^2\log n)$ time).  In this section, we show how to reduce the runtime for finding the matching to $O(n^{1.5}\log n)$ using a data structure containing the blocker shapes from the matching graph,  which hence decreases the runtime for {\sc Placement} to $O(n^2\log n)$.

We follow the idea that 
Cabello \cite{cabello2007approximation} used (in his Lemma 7) that speeds up each phase of the Hopcroft and Karp algorithm to $O(n\log n)$ time. 
We can follow Cabello's approach exactly if we have a data structure for blocker shapes that satisfies the following properties.

\begin{enumerate}
    \item Inserting all blocker shapes into the data structure takes linear time. 
    \item Inserting or deleting a blocker into the data structure takes $O(\log n)$ time.
    \item Given an input rectangle $r$, return a \emph{witness, i.e.,} any blocker shape $b$ in the data structure that touches the rectangle, otherwise return none, in $O(\log n)$ time.
\end{enumerate}

In Cabello's algorithm, blocker shapes are just points, and the input shapes are squares, so Cabello uses the orthogonal range query data structure found in \cite{datastruct}. 


Our blockers are not points, so instead we need a different data structure. Our data structure wraps around two of the orthogonal range query data structures from \cite{datastruct}, call them $\mathcal{D}_1$ and $\mathcal{D}_2$. To insert a blocker $b$, we insert the topmost and rightmost grid points of the blocker shapes into $\mathcal{D}_2$, and insert the rest of $b$'s grid points into $\mathcal{D}_1$. To delete $b$, delete all $O(1)$ grid points from $\mathcal{D}_1$ and $\mathcal{D}_2$. Initializing can be done by initializing  $\mathcal{D}_1$ and $\mathcal{D}_2$ with the appropriate grid points. Initialization takes time linear in the number of blockers inserted initially, which will be $O(n^2)$. Insertion and deletion take $O(\log n)$ time.


The final question is how to find a witness blocker shape in our data structure that intersects the given input rectangle $r$. First, if the rectangle is small, return none, as these rectangles are not considered to be touching any blockers. Otherwise, if $r$ intersects a grid point of a blocker in the data structure, we can just query $\mathcal{D}_1$ and $\mathcal{D}_2$ for a grid point, and return the blocker it belongs to. If $r$ only intersects a segment of a blocker $b$ in the data structure, let's say it is the vertical edges of the rectangle hitting a horizontal blocker segment without loss of generality. Then round the vertical edges of $r$ down to the nearest grid edge, call this rectangle $r'$. $r'$ will hit the left or middle or bottom grid point of $b$, so we can query $\mathcal{D}_1$ with $r'$. Similarly, round the horizontal edges of $r$ down and perform a query on $\mathcal{D}_1$ with this rounded rectangle. If any of the queries fail, then there can be no witness. In particular, if a rounded rectangle intersects a left/middle/bottom grid point of blocker $b$, then the original $r$ must have intersects $b$. This operation takes $O(\log n)$ time.


With this data structure, the exact same argument as in Cabello's proof of Lemma 7 holds. We will briefly reiterate the idea of the proof. At the start of each phase of the Hopcroft and Karp matching algorithm, we require the data structure contain all blocker shapes from the bipartite matching graph $H^-$. The blockers are inserted into the data structure in $O(n^2)$ time before the first phase, and later we'll describe how to reset the data structure at the end of each phase. Next, note the observation that the number of vertices in all layers of the layered graph computed by the Hopcroft and Karp algorithm, excluding the last layer with at least one exposed vertex, has at most $2n$ vertices, as there are at most $n$ rectangles and at most $n$ matched blocker shapes.
Following Cabello's argument, we can compute the layered graph without the last layer using a modified version of breadth first search in $O(n\log n)$ time. Edges adjacent to a rectangle vertex are found by repeatedly querying for a witness blocker shape until none is returned, then deleting the witness from the data structure, and using the graph edge from the rectangle to the blocker shape. This is opposed to constructing the graph's adjacency list to perform breadth first search, as there are too many edges to do breadth first search in the usual way. We will end up partially constructing the last layer, but will stop after seeing the first exposed (not matched) blocker. The data structure will not contain any vertices from before the last layer, and all other vertices remain (at most $n$ blocker vertices from the last layer may have been deleted earlier and can be reinserted). Then, the vertex disjoint augmenting paths are computed using depth first search on the layered graph, but using the data structure to find edges towards the last layer. Deleting exposed vertices (witness blocker shapes) from the data structure ensures the vertices are not in two different augmenting paths. Only $O(n)$ queries will be used in this step, as there are at most $O(n)$ vertices in the second last layer. 

Each step of each phase takes $O(n\log n)$ time, as $O(n)$ insertions, deletions and witness queries are performed. Each phase must start with the data structure being full, so we reinsert all deleted blocker vertices in $O(n\log n)$ time. With $\sqrt{n}$ phases, the total runtime to find the matching is $O(n^{1.5}\log n).$ This runtime is ignoring the time to initialize the data structure, but this time is not the bottleneck either. Therefore, the runtime for {\sc Placement} is $O(n^2 \log n)$ time.

\section{Proof of Lemma~\ref{lemma:n-cubed}}
\label{app:candidates}

In this section, we prove Lemma~\ref{lemma:n-cubed}, which states that for the $L_\infty$-distance the optimal value $\delta^*_\infty$ takes one of $O(n^3)$ possible values.  More specifically, 
$\delta_\infty^*=(t-b)/k$ where $k \in \{1,\dots,n\}$ and either $t$ is the top coordinate $t(R)$ of some rectangle  $R$ and $b$ is the bottom coordinate $b(R')$ of a different rectangle $R'$, or $t$ is the right $r(R)$ coordinate of some rectangle $R$ and $b$ is the left coordinate $\ell(R')$ of a different rectangle $R'$.

To prove this, without loss of generality, assume that all rectangle corners have
non-negative coordinates.  Among all optimal solutions, take the one
that minimizes $\sum_R (x(R)+y(R))$, where $(x(R),y(R))$ is the representative of $R$.
For any rectangle $R$, then
either $x(R)=\ell(R)$ or $x(R)=x(R')+\delta^*_\infty$
for some other rectangle $R'$.  For if neither were the case, then the point
$(x(R){-}\varepsilon,y(R))$ (for sufficiently small $\varepsilon$) would also 
lie within $R$, and have distance $\delta^*_\infty$ or more from all other
points.  

In consequence, we can write $x(R)=\ell(R')+k\delta^*_\infty$
for some integer $k\in \{0,\dots,n\}$, where $R'$ is some other
rectangle (possibly $R=R'$).
Namely, either $x(R)=\ell(R)$ (then $k=0$) or $x(R)=x(R')+\delta^*_\infty$ for some rectangle $R'$,
in case of which $x(R')<x(R)$ and by induction (on the rank of $R$ with respect to $x(R)$) we have $x(R')=\ell(R'')+k'\delta_\infty^*$ for some $R'',k'$
and so $x(R)=\ell(R'')+(k'+1)\delta_\infty^*$.

In a completely symmetric way we can show that for any rectangle $R$ we 
have $y(R)=b(R')+k\delta_\infty^*$ for some rectangle $R'$ and $k\in \{0,\dots,n\}$.
Now we have three cases.  Assume first that for some rectangle $R$ we have $x(R)=r(R)$ 
and $x(R)=\ell(R')+k\delta_\infty^*$ for some $k>0$ and rectangle $R'$.  Then the 
claim holds since $\delta_\infty^*=(x(R)-\ell(R')/k = (r(R)-\ell(R'))/k$.    
Assume next that for some $R,R'$ we have $y(R)=t(R)$ and
$y(R)=b(R')+k\delta_\infty^*$ for some $k>0$.  Then again the claim holds.

Now assume neither of the above cases holds.  We show that this contradicts 
maximality of $\delta_\infty^*$.  Define a new solution $(x'(R),y'(R))$ by 
choosing a sufficiently small $\varepsilon$ and essentially scaling the 
solution by $1{+}\varepsilon$ in both directions.  However, we must be careful 
to ensure that this is a solution.  Formally, if $x(R)\neq r(R)$, then set 
$x'(R)=(1{+}\varepsilon)x(R)$; otherwise keep $x'(R)=x(R)$.  Proceed symmetrically 
for $y'(R)$.  Clearly this is a set of representatives (for small enough 
$\varepsilon$), and we claim that its distance exceeds $\delta_\infty^*$.    
To see this, consider two rectangle $R,R'$ whose representatives had distance 
exactly $\delta_\infty^*$ in the first solution.  (For all other pairs of 
rectangles, the distance continues to exceeds $\delta_\infty^*$ if 
$\varepsilon$ is chosen sufficiently small.)  Because we are in the $L_\infty$-distance, we know that $d(R,R')$ is achieved in one of the two cardinal directions, 
so (say) $x(R)=x(R')+\delta_\infty^*$.  This 
implies $x(R)=x(R'')+k\delta_\infty^*$ for some $k>0$ and some rectangle $R''$.  Therefore
$x(R)\neq r(R)$, otherwise we would have been in the first case.  
So $x'(R)=(1{+}\varepsilon)x(R)=(1{+}\varepsilon)(x(R')+\delta_\infty^*) \geq x'(R')+(1{+}\varepsilon)\delta_\infty^*$,
the desired contradiction.

\section{Proofs of Lemmas~\ref{lem:closed=on-right}, \ref{lemma:detect-critical}, and~\ref{lemma:critical-bits-bound}}
\label{app:proofs}

In this section, we give three omitted proofs.  Recall that Observation~\ref{obs:events} characterized events, one of which
must have happened if {\sc Placement} goes from success to failure.  We restate this observation here for ease of reference.

\begin{observation}
Assume {\sc Placement}($\delta$) succeeds but {\sc Placement}($\delta'$) fails for some $\delta'>\delta$.  Then at least one of the following
events occurs as we go from $\delta$ to $\delta'$:
\begin{enumerate}
\item the set of small/big rectangles changes,
\item the distance between the centres of two small rectangles equals $\hat{\delta}$ for some $\delta\leq \hat{\delta}<\delta'$,
\item the set of blocker-shapes owned by a small rectangle increases,
\item the set of blocker-shapes intersecting a big rectangle decreases.
\end{enumerate}
\end{observation}

To justify this, note that if event (1) does not happen, then the set of big rectangles is the same at $\delta$ and at $\delta'$. The big rectangles form one side of the graph on which the matching algorithm is run.  Furthermore, if events (3) and (4) do not happen, then the graph can only gain edges as we go from $\delta$ to $\delta'$.  Also, nothing changes with respect to event (2).  Thus if {\sc Placement} succeeds at $\delta$ it will also succeed at $\delta'$.  

Observe that coordinates of grid points (hence of blocker shapes) are linear in $\delta$ and hence change continuously over time; we think of them as
``shifted'' (and for blocker shapes, ``scaled'') as we change $\delta$.  Also recall that both rectangles and blocker-shapes are closed.  Therefore we have:

\begin{observation}
\label{obs:big-shift}
If a 
rectangle $R$ does not intersect a blocker-shape $B$ at time $\delta$ then
$R$ does not intersect (the shifted and scaled) $B$ in a neighbourhood of $\delta$.
\end{observation}

\begin{proof}[Proof of Lemma~\ref{lem:closed=on-right}]
To prove that the intervals where {\sc Placement} succeeds are closed on the right, we will show that the intervals where it fails are open on the left.
Consider any value $\delta'$ such that {\sc Placement}$(\delta')$ fails.
We must prove that {\sc Placement} fails at  
$\delta := \delta' - \epsilon$
for sufficiently small $\eps$ (determined below). We will apply 
Observation~\ref{obs:events}.
We must show that 
none of events (1)-(4) can happen between $\delta$ and $\delta'$:
\begin{description}
\item[Event (1).]
We claim that the set of big rectangles is the same throughout.  Namely, if a rectangle is big at time $\delta'$ then (by our complicated definition)
it is also big at all times between $\delta$
and $\delta'$ presuming $\eps$ was chosen smaller than the $\eps_0$ in our definition of big rectangles.
Now suppose rectangle $R$ is small at $\delta'$.  This means that either $R$ does not intersect any blocker-shape at $\delta'$ (in which case, by Observation~\ref{obs:big-shift}, $R$ is small at $\delta' - \epsilon$ for sufficiently small $\epsilon$) or
the boundary of $R$ contacts a blocker-shape, but shifting to $\delta'- \epsilon$ for any $\epsilon>0$ removes the contact (in which case, $R$ is small at $\delta' - \epsilon$).

\item[Event (2).] Since there are $\Theta(n^2)$ values of $\hat{\delta}$ where the distance between rectangle-centres equals $\hat{\delta}$, we can choose $\eps$ small enough such that no such value falls between $\delta'{-}\eps=\delta$ and $\delta'$.

\item[Event (3).] Recall that
a small rectangle $r$ owns blocker-shape $B$ at time $\delta$ if some point of $B$
belongs to the  \emph{open} ball of the appropriate norm and radius 
around the centre point $p(r)$ of rectangle $r$.   Since $B$ changes continuously with time, therefore for sufficiently small $\eps$ no such event occurs 
in the interval $(\delta'-\eps,\delta')$.


\item[Event (4).]
By Observation~\ref{obs:big-shift}
the set of blocker-shapes intersected by $R$ can only increase in the interval.
\end{description}

We also claimed that a critical value $\delta$ provides an $f_\ell$-approximation.  
This holds because {\sc Placement}($\delta+\eps$) fails for all small $\eps>0$, hence $\delta_\ell^*<f_\ell(\delta+\eps)$ for all $\eps>0$ and so $\delta_\ell^* \leq f_\ell\delta$.
\end{proof}

\begin{proof}[Proof of Lemma~\ref{lemma:detect-critical}]
To test if $\delta$ is a critical value, we run
{\sc Placement} at $\delta$ (we want it to succeed) and then run {\sc Placement}  ``symbolically''
at $\delta + \epsilon$ where $\epsilon$ an infinitesimal. 
The idea is that when the algorithm performs any test that depends on $\delta$, we see how the test behaves at $\delta + \eps$ instead.  Alternatively, we can use the bounds of Lemma~\ref{lemma:critical-bits-bound} to get a specific small-enough value of $\eps$ to use.
\end{proof}

\begin{proof}[Proof of Lemma \ref{lemma:critical-bits-bound}]
If $\delta$ is a critical value, then, by definition, {\sc Placement($\delta$)} succeeds and {\sc Placement($\delta + \epsilon$)} fails for any sufficiently small $\epsilon > 0$.    
By Observation~\ref{obs:events} there must have been an event as we go from $\delta$ to $\delta+\eps$, and since $\eps$ is arbitrarily small the event must have been at $\delta$.

Our plan is to show that for any of these events, $\delta$ (or $\delta^2$ in the case of $L_2$) is a rational number with bounded numerator and denominator.
We group events (1) and (4) together and deal with the following three cases:


\remove{ 
We complete this section by giving the proofs of Lemmas~\ref{lem:closed=on-right}, \ref{lemma:detect-critical}, and~\ref{lemma:critical-bits-bound}.


The lemmas use 
the idea of continuously transforming from one grid to another. 
First, for a given value $\delta$, the grid points for the {\sc Placement} algorithm are $\mathcal{G}(\delta) = \{ \gamma_\ell(\delta)  \cdot p : p\in \mathbb{Z}^2 \}$, recall that $\gamma_\ell = \gamma_\ell(\delta) =\frac{1}{2^{1/\ell}}\delta$.  Every point $p \in \mathbb{Z}^2$ has a corresponding point $\gamma_\ell(\delta)\cdot p \in \mathcal{G}(\delta)$, and we can think of $p$ moving continuously as a function of $\delta$. 
Combinatorially, instead of thinking of multiple calls to the {\sc Placement} algorithm as independent, we think of grid points and blocker shapes as objects that move continuously as we increase or decrease $\delta$ (blocker shapes will also continuously stretch or contract, as the distance between grid points increases). We will refer to 
$\delta$ as a ``time" variable, and we can analyze how any given blocker shape ``moves through time" as we increase or decrease $\delta$. With this in mind, we can think of the matching graph $H$ as adding or losing blockers shapes at certain values of $\delta$, or some other change happens, which the definition of \emph{critical} values  tries to capture.

\begin{proof}[Proof of Lemma~\ref{lem:closed=on-right}]
We only need to show that the intervals where Placement succeeds are closed on the right.  It follows immediately that any critical value provides an $f_\ell$ approximation.

We prove that intervals of values where Placement fails are open on the left, i.e., if Placement fails for $\delta$ then it fails for $
\delta - \epsilon$ for sufficiently small $\epsilon > 0$. 
So suppose Placement fails at $\delta$.
We examine each step of the algorithm where success/failure is decided.

\begin{description}
\item[Step 1.] We claim that
\changed{for $\epsilon$ sufficiently small} the set of big rectangles is the same for $\delta$ as for $\delta - \epsilon$. One direction is ensured by our definition of big rectangles: if $R$ is big at $\delta$ then by definition, it is big at $\delta - \epsilon$. 
\changed{For the other direction, we use this observation:}

\begin{observation}
\label{obs:big-shift}
If a \attention{\st{big}} rectangle $R$ does not intersect a blocker-shape $B$ at $\delta$ then, since $R$ and $B$ are closed, $R$ does not intersect (the shifted) $B$ in a \changed{small enough} neighbourhood of $\delta$.  
\end{observation}

\changed{So suppose rectangle $R$ is small at $\delta$.  This means that either $R$ does not intersect any blocker-shape at $\delta$ (in which case, by the observation, $R$ is small at $\delta - \epsilon$ for sufficiently small $\epsilon$) or
the boundary of $R$ contacts a blocker-shape, but shifting to $\delta - \epsilon$ for any $\epsilon>0$ removes the contact (in which case, $R$ is small at $\delta - \epsilon$).
}


\item[Step 3.] If this step fails at $\delta$ then two centres of small rectangles have distance $< \delta$.  Then at $\delta - \epsilon$ (where the rectangles are still small) their centres have distance $< \delta - \epsilon$ for sufficiently small $\epsilon$, so this step fails at $\delta - \epsilon$.

\item[Step 4.] We claim that 
if small rectangle $r$ owns blocker-shape $B$ at $\delta$ then $r$ owns the (shifted) blocker-shape $B$ at $\delta - \epsilon$ for $\epsilon$ sufficiently small. 
This just depends on the definition of ``ownership'' via the  \emph{open} ball of the appropriate norm and radius 
around the centre point $p(r)$ of rectangle $r$. 
More precisely, if there is a point $q$ on $B$ within this open ball then, if $\epsilon$ is sufficiently small, the shift of $B$ (i.e., the shift of point $q$) plus the shrinking of the ball's radius will leave 
the shifted point $q$ still  within the shrunken ball around $p(r)$, so $r$ still owns the blocker-shape at $\delta - \epsilon$.

\remove{
$d_\ell(p(r),q) < \delta$, say $d(p(r),q) =  \delta - \tau$ then, if $\epsilon$ is very small, then the shift of $B$ (i.e., the shift of point $q$) plus the shift of the ball to radius $\delta - \epsilon$ will be less than $\tau$ so the shifted point $q$ will still be within $\delta - \epsilon$ of $p(r)$ so $r$ still owns the blocker-shape at $\delta - \epsilon$.
}

This claim implies that the set of owned blocker-shapes can only increase from $\delta$ to $\delta - \epsilon$ so the bipartite graph $H$ does not gain any edges because of this step. 

\item[Steps 5--8.] Observation~\ref{obs:big-shift} above implies that 
 for each big rectangle, the set of blocker-shapes it intersects can only decrease from $\delta$ to $\delta - \epsilon$.  Thus, if the matching algorithm fails at $\delta$ it will also fail at $\delta - \epsilon$.
\end{description}

This completes the proof of Lemma~\ref{lem:closed=on-right}.
\end{proof}

\begin{proof}[Proof of Lemma~\ref{lemma:detect-critical}]
To test if $\delta$ is a critical value, we run 
Placement at $\delta$ (we want it to succeed) and then run Placement  ``symbolically'' 
at $\delta + \epsilon$ where $\epsilon$ an infinitesimal. 
\anna{Say more?} 
\end{proof}



{\bf Proof of Lemma \ref{lemma:critical-bits-bound}}
If $\delta$ is a critical value, then, by definition, the Placement algorithm succeeds at $\delta$ and fails at $\delta + \epsilon$ for any sufficiently small $\epsilon > 0$.
This can only happen if one of the following four events occurs between $\delta$ and $\delta + \epsilon$:
\begin{enumerate}
\item the set of small/big rectangles changes,
\item the distance between the centres of two small rectangles is exactly $\delta$,
\item the set of blocker-shapes owned by a small rectangle changes,
\item the set of blocker-shapes intersecting a big rectangle changes.
\end{enumerate}

To justify this, note that if none of these events occur, then the matching algorithm is run on the same graph, so its output will not change.

Our plan is to show that for any of these events, $\delta$ (or $\delta^2$ in the case of $L_2$) is a rational number with bounded numerator and denominator.
We group events (1) and (4) together and deal with the following three cases:

} 

\begin{description}
\item[Case A.] Event (1) or (4). 
These events only occur when a line of the blocker-shape grid becomes coincident with a line of the $D \times D$ grid of the input rectangles, i.e., when $i \gamma_\ell = k$, for some integer $k$, $1 \le k \le D$, and some index $i$, $ 1 \le i \le D/\gamma_\ell$.

In $L_1$, we have $\gamma_1 = \delta/2$, so $\delta = 2k/i$.  The numerator is $2k$ which is $\le 2D$.  The denominator is $i$ which is $\le D/\gamma_1 = 2D / \delta \le 2Dn$, where we use the fact that $\delta \ge 1/n$ (Claim~\ref{claim:delta-bounds}).  

In $L_2$, we have $\gamma_2 = \delta / \sqrt{2}$, so $\delta^2 = 2k^2 / i^2$.  Following the same steps as above, the numerator is $\le 2D^2$ and the denominator is $\le 2D^2 n^2$. 


\item[Case B.] Event (2).  $\delta = d_\ell (c_1,c_2)$ where $c_1$ and $c_2$ are the centres of two rectangles, i.e., the average of two opposite corners.  In $L_1$,  $\delta$ has a numerator $\le 2D$ and denominator $2$, and in $L_2$, $\delta^2$ has a numerator $\le 2D^2$ and denominator $4$.

\item[Case C.] Event (3).  Recall that the algorithm decides ownership using $L_1$ distance even when we are working in $L_2$.  
Thus, event (3) only occurs when the centre point of a rectangle lies on the boundary of the $L_1$-ball 
$C'$ shown by the cyan boundary in Figure~\ref{fig:blocker-ball}(left).  This means that  
a diagonal line of slope $\pm 1$ through points of the blocker-shape grid becomes coincident with a half-integer point, i.e., $(i \pm j) \gamma_\ell  = k/2$, for some integer $k$, $1 \le k \le 2D$, and indices $i, j$,  $ 1 \le i \le D/\gamma_\ell$.

In $L_1$, we have $\gamma_1 = \delta/2$,
so $\delta = k/(i \pm j)$.  The numerator is $\le 2D$ and the denominator is $\le 4Dn$ by the same analysis as in Case A.

In $L_2$, we have $\gamma_2 = \delta / \sqrt{2}$, so $\delta^2 = k^2/(i \pm j)^2$.  The numerator is $\le D^2$ and the denominator is $\le 8D^2n^2$.
\end{description}
This completes the proof of Lemma~\ref{lemma:critical-bits-bound}.
\end{proof}

\section{Further details on hardness results}
\label{app:hardness}

\remove{ 
\anna{remove duplicate stuff}
In this section we give NP-hardness and 
APX-hardness 
results for the 
distant representatives problem.
We first show that, even for the special case of unit horizontal segments, the decision version of the problem is NP-complete for the $L_1$ and $L_\infty$ norm, and NP-hard for the $L_2$ norm. (For the $L_2$ norm, bit  complexity issues prevent us from placing the problem in NP.) 
The NP-hardness result for unit horizontal intervals in the $L_2$ norm was proved previously by Roeloffzen in  his Master's thesis~\cite[Section 2.3]{roeloffzenfinding}. 
However, he did not take care to bound  the bit complexity of the coordinates of the interval endpoints that he constructed in his  reduction.  
Our reduction for $L_2$ is similar to his, and we fill in these missing details.

Next, we enhance our reductions to ``gap-producing reductions'' 
 to obtain lower bounds on the approximation  factors that can be  achieved in polynomial time.  Here, our goal is to compare with 
 our approximation algorithms
 for the case  of rectangles, so rather than  considering the special case of unit horizontal segments, we consider the more general case of horizontal and vertical segments in  the plane. 
Our main result is that,
assuming P $\ne$ NP, no polynomial time approximation algorithm achieves a factor better than
$1.5$ in $L_1$ and $L_\infty$ and $1.4425$ in $L_2$.

Our reductions are from Monotone Rectilinear Planar 3-SAT,
which is NP-complete by a result of de~Berg and Khosravi~\cite{deberg2010optimal}.
In this version of 3SAT, ``monotone'' means that every clause has either three positive literals or three negative literals.  ``Rectilinear planar'' means that we are given a representation by non-crossing line segments in the plane where each variable is represented by a vertical segment at $x$-coordinate $0$, each positive [negative] clause is represented by a vertical segment at a positive [negative, resp.] $x$-coordinate, and there is a horizontal line segment joining any variable to any  clause that contains it.  See Figure~\ref{fig:3-SAT-plan}(a) for an example instance of the problem.  Note that the representation is rotated $90^\circ$ from the original in~\cite{deberg2010optimal}, and the vertical segments corresponding to variables and  clauses are thickened to form constant width rectangles. For an instance of Monotone Rectilinear Planar 3-SAT with $n$ variables and $m$ clauses, the representation can be on an $O(m) \times O(n+m)$ grid.

\begin{figure}[tb]
    \centering
    \includegraphics[width=\textwidth]{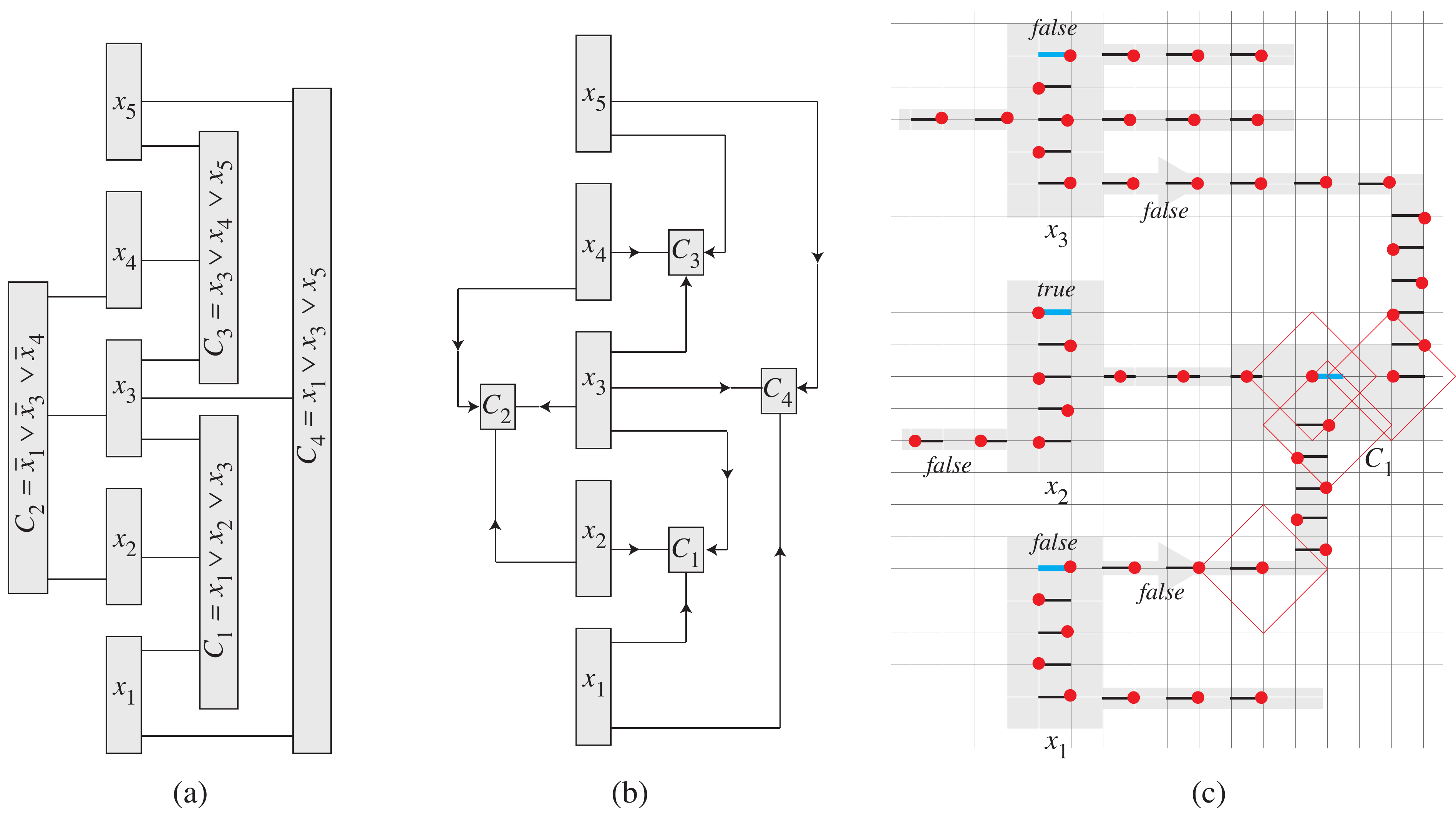}
    \caption{(a)  An instance of Monotone Rectilinear Planar 3-SAT. (b) The corresponding ``Short Clause Representation'' that gives our layout of variable gadgets, wires, and clause gadgets. (c) A detail of our construction for clause $C_1 = x_1 \vee x_2 \vee x_3$ in the $L_1$ norm
    showing how the truth-value setting $x_1 = $ False, $x_2 = $ True,  $x_3 = $ False, permits representative points (shown as red dots) at distance at least $\delta_1 = 2$.
}
    \label{fig:3-SAT-plan}
\end{figure}



\  




} 

\subsection{ NP-hardness results }
\label{sec:appendix-NP-hardness}

\remove{ 
\begin{theorem}
The decision version of the distant representatives problem for unit horizontal segments in the $L_1, L_2$ or $L_{\infty}$ norm is NP-hard.  
\end{theorem}

In Section~\ref{sec:bit-complexity} we show that the decision problems in fact lie in NP for the 
$L_1$ and $L_\infty$ norms even for the   more general   case of rectangles, and we discuss the bit complexity issues that prevent us from placing the decision problem in NP  for the $L_2$  norm.

As noted above, the NP-hardness result for unit horizontal intervals in the $L_2$ norm was proved previously by Roeloffzen in  his Master's thesis~\cite[Section 2.3]{roeloffzenfinding}.
However, he did not take care to bound  the bit complexity of the coordinates of the interval endpoints that he constructed in his  reduction.  
Our reduction for $L_2$ is similar to his, and we fill in these missing details. The core issue is that we need to construct (many) pairs of points $p$ and $q$ in the plane and bound the number of bits in their coordinates and in their $L_2$ distances
which requires the use of 
scaled Pythagorean triples.

Our reduction is from Monotone Rectilinear Planar 3-SAT.  
We first modify the representa\-tion 
to a ``Short Clause Representation''
where each clause rectangle has fixed height and is connected to its three literals via three ``wires''---the middle one remains horizontal, the bottom one bends to enter the clause rectangle from the bottom, and the top one bends twice to enter the clause rectangle from the far side.  See Figure~\ref{fig:3-SAT-plan}(b).
Each wire is directed from the variable to the clause, and represents a literal. 
The representation is still on an 
$O(m) \times O(n+m)$ grid.

To complete the  reduction to the distant representatives problem we 
replace the rectangles with variable and clause gadgets constructed from unit horizontal intervals, and also implement the  wires using such intervals.  
The details 
depend on the norm $L_\ell$, $\ell  = 1, 2, \infty$.
We also set a value of $\delta_\ell$ to obtain a decision problem that asks if  there is an assignment of a representative point to each interval that is \emph{valid}, i.e., such that no two points are closer than $\delta_\ell$.  
We set $\delta_1 = 2$, $\delta_2 = \frac{13}{5}$, and $\delta_\infty = \frac{1}{2}$.
An example of the construction for $L_1$ (with $\delta_1 = 2$) is shown in Figure~\ref{fig:3-SAT-plan}(c).

} 

Our constructions will ensure the following properties.

\begin{enumerate}[label=\textbf{P\arabic*},ref=P\arabic*]
    \item \label{prop:variable}
    Each variable gadget has a \emph{variable interval} whose  representative point can only be at its left endpoint (representing the True  value of the variable)  or at its right endpoint (representing the False value of the variable).
    With either choice, there is a valid assignment of representative points to all intervals in the variable gadget.   

    
    \item \label{prop:wire} 
    A wire is constructed as a sequence of intervals. 
    There are two special valid assignments of representative points to the intervals of a wire which we call the ``false setting'' and the ``true setting''. 
    Details will be given later on, but for now, we just note that 
    along the horizontal  portion of a wire, the false setting places the representative point of each interval at the interval's forward end (relative to the  direction  of  the wire).  See Figure~\ref{fig:3-SAT-plan}(c).
    
    The true/false settings for wires behave as follows. 
    If a wire corresponds to a literal that is false (based on  the left/right  position of the representative  point of the corresponding variable interval),  then the false setting is the \emph{only} valid assignment of representative points for the  intervals in  the   wire.
    If a wire corresponds to a literal that is true, then the true setting of the wire is \emph{a} valid assignment of representative points. 
    Note the   asymmetry here---the false setting is forced but the true setting is not.

    \item \label{prop:clause} Each clause gadget consists of one \emph{clause interval}.
    If all three wires coming in  to a clause gadget have the false setting then there is no valid assignment of a representative  point to the clause interval.  If at least one of the wires has a true setting then there is a valid assignment of a representative  point to the clause interval.
    
\end{enumerate}

\begin{lemma}
Any construction with the above properties gives a correct reduction from Monotone Rectilinear Planar 3-SAT to the decision version of the distant representatives problem.
\label{lem:propertiesLemma}
\end{lemma}
\begin{proof}
We must prove that the  original instance, $\Phi$, of Monotone Rectilinear 3-SAT is satisfiable if and only if the constructed instance, $\mathcal I$, of the distant representatives problem has a valid assignment of representative points, i.e., an assignment of representative points such that any  two points are at least  distance $\delta_\ell$ apart.  

First suppose that $\Phi$ is satisfiable. By  Property~\ref{prop:variable} we can  choose a 
a valid assignment  of representative points to the intervals of the variable gadgets such that a variable being True/False corresponds to using the  left/right  endpoint (respectively) of  the variable interval. 
We then choose the true/false settings  of the wires  according to Property~\ref{prop:wire}---the  false setting  for wires of false literals and the true setting for wires of true literals. 
Since $\Phi$ is satisfiable, every clause contains a True literal. The corresponding  incoming wire has been given a true setting. Then,
by Property~\ref{prop:clause}, the clause interval has a valid assignment of representative  points. Thus $\mathcal I$ has a valid assignment of representative points.

For the other direction, suppose $\mathcal I$  has a valid assignment of representative points.  By  Property~\ref{prop:variable} this corresponds to a truth value assignment to the  variables. By Property~\ref{prop:wire} the  wires corresponding  to false literals can only have the false setting (though we don't know about the wires corresponding to true literals).   By Property~\ref{prop:clause} every clause has at least one incoming wire that does not have the  false setting, and  this wire must then correspond to a True literal.  Thus, every clause is satisfied and $\Phi$ is satisfiable. 
\end{proof}

In the following subsections we describe the variable gadgets, wires, and clause gadgets  
for each of the  norms $L_1, L_2, L_\infty$.  
In each case we prove that the above properties hold. 

\subsubsection{$L_1$ norm, $\delta_1 = 2$.}

\paragraph*{Variable gadget.} We use a \emph{ladder} consisting of unit intervals, called \emph{rungs}, in a vertical pile, unit distance apart. 
See Figure~\ref{fig:L1-construction}(a).
Number the rungs starting with rung 1 at the top.  Observe that the $L_1$ distance between opposite endpoints of two consecutive rungs  is $\delta_1 = 2$.  Thus there are precisely two ways to assign representative points to a ladder of at least two rungs.  For Property~\ref{prop:variable}, let the variable interval be rung 1, and associate the value True [False] if rung 1 has its representative point on the left [right, resp.].

\begin{figure}[htb]
    \centering
    \includegraphics[width=.7\textwidth]{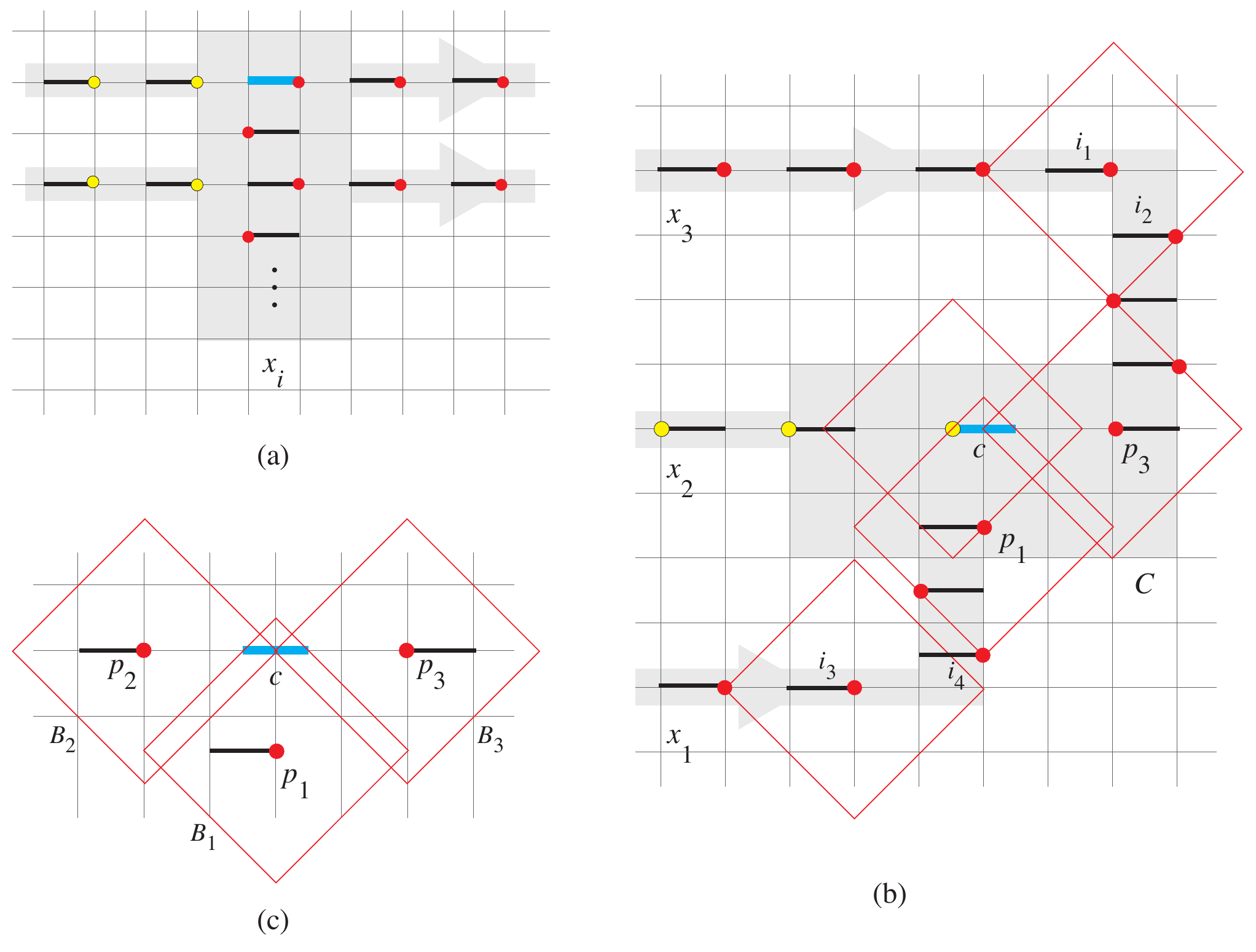}
    \caption{Construction for the $L_1$ norm.  (a) Variable gadget.  Intervals have length 1 and  $\delta_1 = 2$.   Variable $x_i$ is shown with the False setting where the representative point (the red dot) on the variable interval (shown in cyan) is on the right end.  The two wires heading right are forced to have the false setting (shown with red dots).  The two wires heading left have the true setting (shown with yellow dots).
    (b) Wires entering the 
    clause gadget for clause $C = x_1 \vee x_2 \vee x_3$.  The clause interval $c$ (shown in cyan) is displaced horizontally by $1/2$.
    Valid representative  points  are shown for the truth-value setting $x_1 = $ False, $x_2 = $ True, $x_3 = $ False.
    (c) A close-up of the clause interval $c$ showing  the $L_1$ balls $B_i$ of radius $\delta_1 = 2$ centred at $p_i$, $i=1,2,3$. 
    }
    \label{fig:L1-construction}
\end{figure}

\paragraph*{Horizontal wire.} For each horizontal  portion  of a wire, use a sequence of unit intervals separated by gaps of length 1. Attach the wires to the odd numbered rungs of a ladder in a variable gadget, with the rung acting  as the first interval of the wire. The false  setting has representative  points at the forward end of each interval (relative to the direction of the wire).  The true setting has representative points at the other end of each interval. 
For Property~\ref{prop:wire} 
(that the false setting is forced) 
observe that 
if a variable is False then its odd-numbered rungs have their  representative  points on the right, so  any horizontal wire extending to the right is forced to use representative points on the right (the forward end) of every interval of the wire.  
On  the other hand, if a variable is True then horizontal wires extending to the right \emph{may} use the true setting. 
Analogous properties hold for the horizontal wires extending to the left.

\paragraph*{Turning wires.}
We focus on the situation for a positive clause---the situation for a negative clause is symmetric.
The top wire coming in to a clause gadget turns downward via a wire ladder as shown  in Figure~\ref{fig:L1-construction}(b). Note that the false setting of interval $i_1$ in the  figure forces the false setting  of interval $i_2$, which then forces the settings down the wire ladder.  Note that the wire ladder can be as long as needed.  Since wires emanate from odd-numbered rungs of variable ladders, the wire ladder has an  even number of rungs and 
the bottom  interval of the wire ladder, at the horizontal line of the middle wire coming in to the clause, has its false setting on the left (see  point $p_3$ in the figure).  
One can verify that the true setting (with representative points at the opposite end  of each interval)  is  valid.

The bottom wire coming in to  a clause gadget turns  upward as shown  in Figure~\ref{fig:L2-construction}(b) via a wire ladder of intervals that are on the half-grid.
This wire ladder has an odd number of rungs, and 
we can ensure at least 3 rungs.  
The false setting of interval $i_4$ is forced because of the false setting of interval $i_3$ together with the ladder above $i_4$.  
The topmost interval of the wire ladder has its false setting on  the right.
One can verify that the true setting (with representative points at the opposite end  of each interval)  is  valid.

We have now established Property~\ref{prop:wire} for wires that turn.





\paragraph*{Clause gadget.}
See Figure~\ref{fig:L1-construction}(c).  The figure shows the clause interval $c$ together with the last interval in each of  the three  wires that come in to the clause gadget, and  the false settings of their representative  points at $p_1$, $p_2$, $p_3$.  
The $L_2$ distance between $p_1$ and either endpoint of $c$ is $\delta_1 = 2$.
Let $B_i$ be the $L_1$ ball of radius $\delta_1$ centred at $p_i$, $i=1,2,3$.  
We now verify Property~\ref{prop:clause}.
Observe that no  point of the interval $c$ is outside all three balls. Thus, if all three incoming  wires have the false setting, there is  no valid representative  point for interval $c$.
However, if at least one of the incoming wires has the true setting, we claim that there is a valid representative point on interval $c$:
If $p_1$ is at the left of its interval, use the midpoint of $c$, and if either of $p_2, p_3$ is at the other endpoint of its interval, use the endpoint of $c$ on that side.
Thus  Property~\ref{prop:clause} holds.


\subsubsection{$L_2$ norm, $\delta_2 = \frac{13}{5}$}

Consider two unit intervals one above the  other, separated by vertical distance $d_y$.  The $L_2$ distance between opposite endpoints of the intervals  is $d = \sqrt{1 + d_y^2}$. In order to have rational values for $d_y$ and $d$, we need scaled  Pythagorean  triples, natural numbers $a,b,c$ with $a^2 + b^2 = c^2$.   We base  our construction on the Pythagorean  triple $5, 12, 13$.  (The triple $3,4,5$ causes some interference.)  
To avoid writing  fractions everywhere, we  describe the construction for intervals of  length $5$, with $\delta_2 = 13$.  Scaling everything by $\frac{1}{5}$ gives us back unit intervals. 

Our construction of variable gadgets and horizontal wires is like the $L_1$ case, just with different spacing.

\paragraph*{Variable  gadget.}   
Use a ladder with  rungs of length 5  spaced 12 vertical units apart. 
See Figure~\ref{fig:L2-construction}(a).
The $L_2$ distance between  opposite endpoints of two consecutive rungs is $\delta_2 = 13$.  
Associate the value True [False] if rung 1 has its representative point on the left [right, resp.].  Property~\ref{prop:variable} holds.

\begin{figure}[htb]
    \centering
    \includegraphics[width=.8\textwidth]{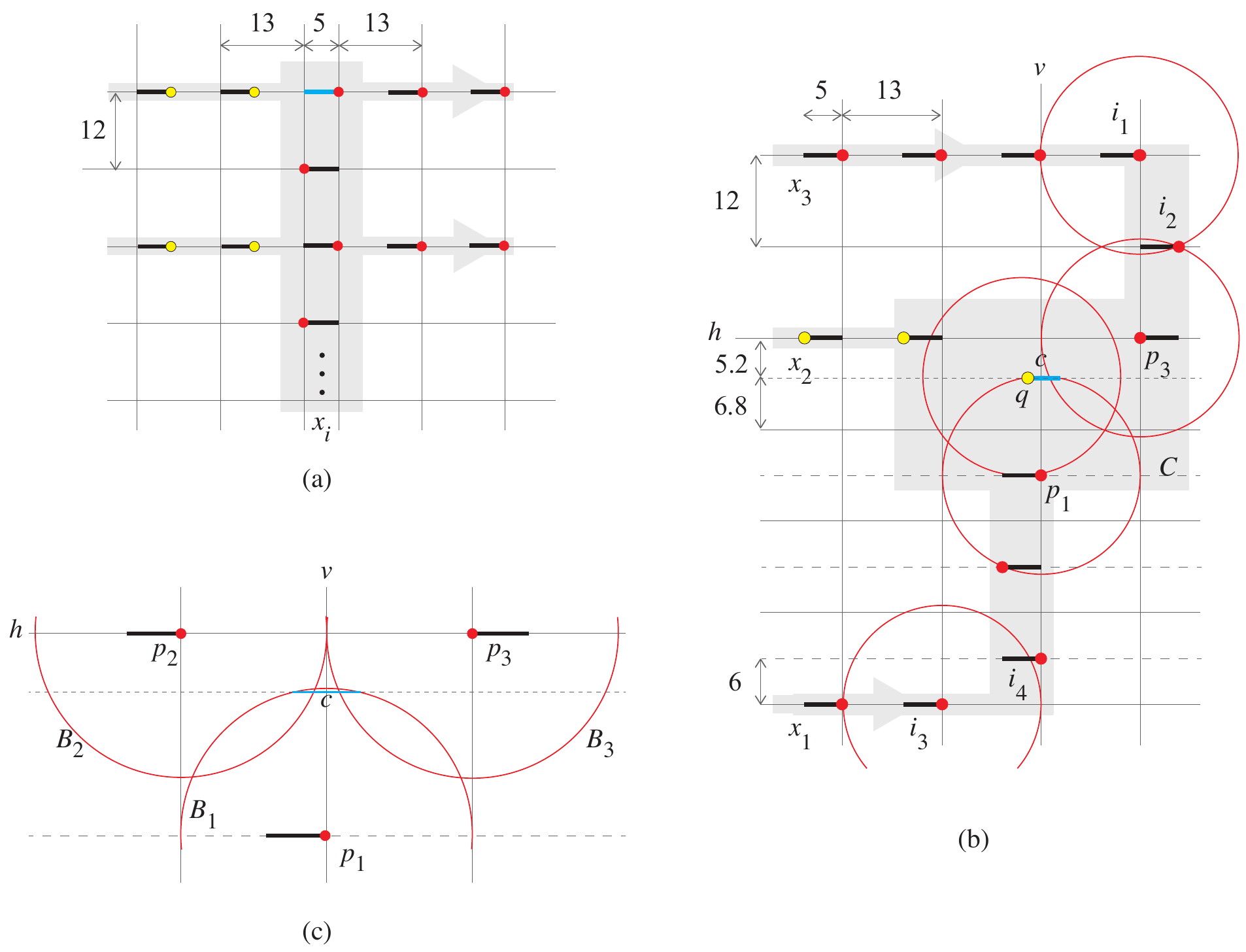}
    \caption{Construction for the $L_2$ norm.  (a) Variable gadget.  Intervals have length 5 and  $\delta_2 = 13$.   Variable $x_i$ is shown with the False setting where the representative point (the red dot) is on the right end of the variable interval (shown in cyan).  The two wires heading right are forced to have the false setting (shown with red dots).  The two wires heading left have the true setting (yellow dots).
    (b) Wires entering the 
    clause gadget for clause $C = x_1 \vee x_2 \vee x_3$.  The clause interval $c$ (shown in cyan) extends from $v - 2.5$ to $v + 2.5$ at $y$-coordinate $h - 5.2$, where $v$ and $h$ are the grid coordinates as shown.
    Valid representative  points are shown for the truth-value setting $x_1 = $ False, $x_2 = $ True, $x_3 = $ False.  
    (c) A close-up of the clause interval $c$ showing  the balls $B_i$ of radius $\delta_2 = 13$ centred at $p_i$, $i=1,2,3$.  
    }
    \label{fig:L2-construction}
\end{figure}

\paragraph*{Horizontal wire.}  For each horizontal portion of a wire, use intervals of length 5 separated by gaps of length 8 (so  the  right endpoints  of two consecutive intervals are distance 13 apart).  Attach the wires to the odd-numbered rungs of the ladder of the variable gadget.  The false setting has representative points at the forward end of each interval, and is forced if the corresponding  literal is False.
The true setting has representative points at the other end of each interval, and is valid if the corresponding literal is True.
So Property~\ref{prop:wire} holds.  

\paragraph*{Turning wires.} 
As for $L_1$, we focus on the situation for a positive clause---the situation for a negative clause is symmetric.
The top wire coming in to a clause gadget turns downward as shown  in Figure~\ref{fig:L2-construction}(b). 
As for $L_1$, the false setting of interval $i_1$ forces the false setting  of interval $i_2$, which then forces 
the bottom  interval of the wire ladder (coming in to the clause) to have its false setting on the left (see  point $p_3$ in the figure).  
One can verify that the true setting (with representative points at the opposite end  of each interval) is  valid.

The bottom wire coming in to  a clause gadget turns  upward as shown  in Figure~\ref{fig:L2-construction}(b) via a wire ladder of intervals that are on the half-grid, i.e., $i_4$ in the figure is 6 units above $i_3$. 
This wire ladder has an odd number of rungs, and 
we can ensure at least 3 rungs.  
The false setting of interval $i_4$ is forced because of the false setting of interval $i_3$ together with the ladder above $i_4$.  
The topmost interval of the wire ladder has its false setting on  the right.
One can verify that the true setting (with representative points at the opposite end  of each interval)  is  valid.

We have now established Property~\ref{prop:wire} for wires that turn.

\paragraph*{Clause gadget.}  See Figure~\ref{fig:L2-construction}(c).  The figure shows the clause interval $c$ together with the last interval in each of  the three  wires that come in to the clause gadget, and  the false settings of their representative  points at $p_1$, $p_2$, $p_3$.  
The clause interval $c$ extends from $v - 2.5$ to $v + 2.5$ at $y$-coordinate $h - 5.2$, where $v$ and $h$ are the grid coordinates as shown. 
Then the $L_2$ distance between $p_1$ and either endpoint of $c$ is  $\sqrt{2.5^2 + 12.8^2} \approx 13.04$ which  is greater  than $\delta_2 = 13$.
Let $B_i$ be the ball of radius $\delta_2$ centred at $p_i$, $i=1,2,3$.  The  endpoints of $c$ lie just outside $B_1$.
We now verify Property~\ref{prop:clause}.
Observe that no  point of the interval $c$ is outside all three balls. Thus, if all three incoming  wires have the false setting, there is  no valid representative  point for interval $c$.  We now consider what happens if at least one incoming wire has the true setting, i.e., if $p_1$, $p_2$, or $p_3$ were at the other end of its interval.  
If $p_2$ were at the other endpoint of its interval, then the left endpoint of $c$ would be a valid representative  point.   Similarly, if $p_3$ were at the  other endpoint of its interval, then the right endpoint of $c$ would be a valid representative  point.  Finally, if $p_1$ were at the  other endpoint of its interval, then the midpoint of $c$ would be a valid representative  point. 
Thus  Property~\ref{prop:clause} holds.

\subsubsection{$L_\infty$ norm, $\delta_\infty = \frac{1}{2}$}

In this case ladders still work, but it is difficult to  attach wires to ladders, so we use a more complicated variable  construction. 
%
%
A further difficulty for the $L_\infty$ case is that we were  unable  to construct a clause gadget of unit intervals based on choosing representative points only on the left/right endpoints of intervals.
(Although this is easy if the clause interval can have length 2.)
Instead, our construction will  place  representative  points at the endpoints or at the middle of each interval, which is why we set $\delta_\infty =  \frac{1}{2}$. 
For this norm, we describe horizontal wires first.

\paragraph*{Horizontal wire.} We use a double row of unit intervals spaced $1/6$ apart vertically.   Specifically, along one  horizontal line, we place a sequence of unit intervals with endpoints at each integer coordinate, and along the horizontal line $1/6$ below, we place a sequence of unit intervals with endpoints at each half integer coordinate.  See Figures~\ref{fig:infty-variable} and~\ref{fig:Linf-construction}.  For Property~\ref{prop:wire}, note that if two consecutive intervals have their representative points at their right endpoints, then all intervals further to the right on the wire must also  have their representative points at their right endpoints. 
Along a horizontal wire, the true setting places representative points at the midpoints of the intervals.

\paragraph*{Variable gadget.}  This gadget has eight intervals
as shown in  Figure~\ref{fig:infty-variable}.
The three intervals $i_1$ are coincident (or almost so), as are the three intervals $i_2$.  These force the representative point on interval $i_3$ to the middle of the interval.  Then the representative point on the variable interval (coloured cyan in the figure) must be 
either the left endpoint (representing a True value) or the right endpoint (representing a False value).

This eight-interval configuration is expanded to a ``double ladder'' with intervals spaced $1/6$ apart vertically as shown in Figure~\ref{fig:infty-variable}.
Down the ladder on the false side (which is the right side in the figure), the representative point 
for each interval is forced by those on the two  intervals above it. 
Down the ladder on the true side (the left side in the figure) we may use the assignment of representative points as shown in the figure. 
Wires extend to the right and left of the double ladder as shown in the figure.  Any wire corresponding to a false literal is forced to the false setting.  Any wire corresponding to a true literal may have the true setting.
Properties~\ref{prop:variable} and~\ref{prop:wire} hold. 

\begin{figure}[htb]
    \centering
    \includegraphics[width=.6\textwidth]{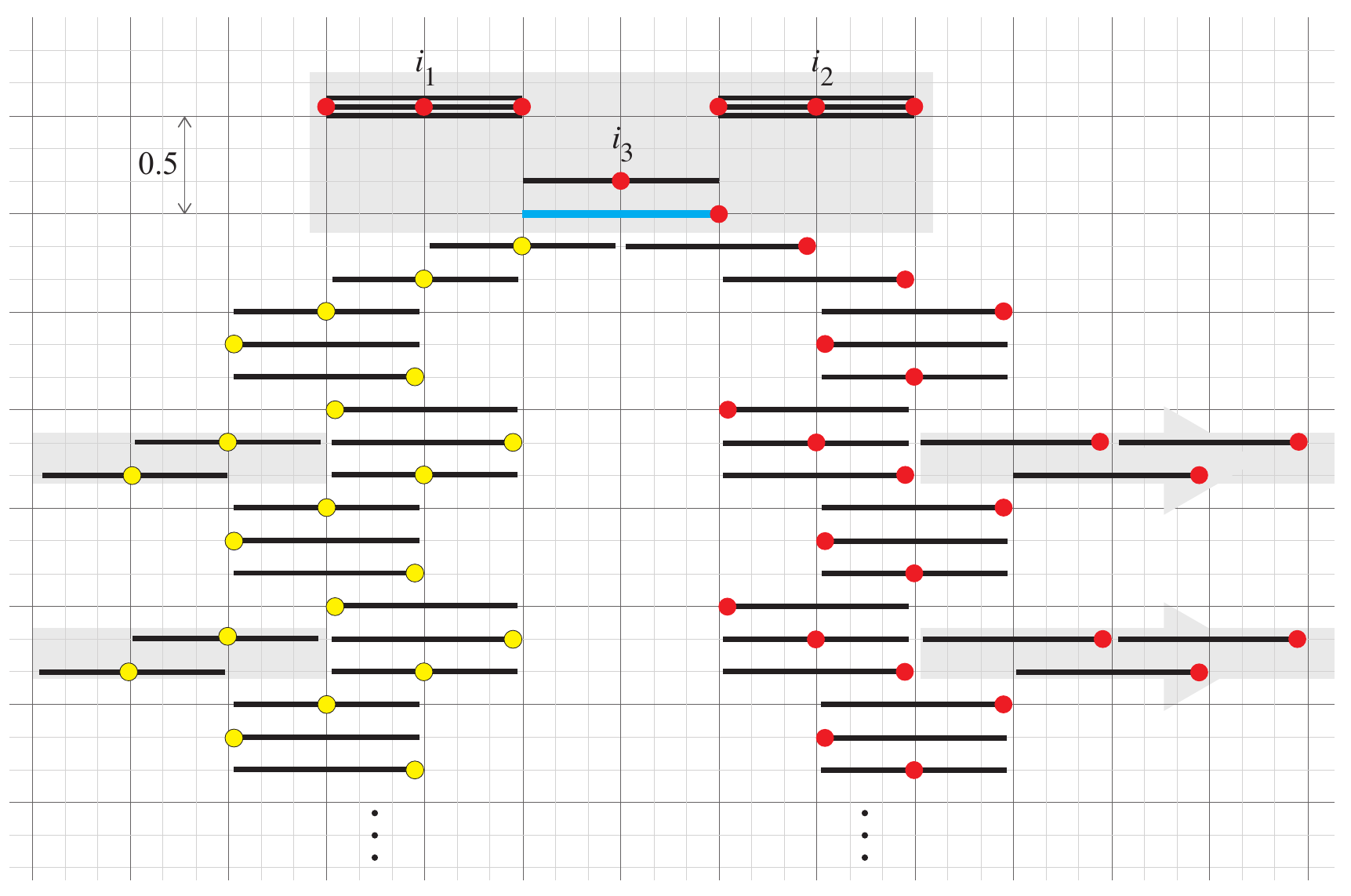}
    \caption{The variable gadget and emanating  wires for $L_\infty$. Grid spacing is $1/6$.  The shaded rectangle at the  top shows  the variable gadget.  The variable interval (shown in  cyan) has a choice of representative point at the left endpoint (True) or the right  endpoint (False).  The  False choice is shown here.  The placement of representative points down the right hand ladder (shown  as red dots) is then forced, and then the wires emanating to the right are forced to the false setting.   The placement of points down the left hand ladder (shown as yellow dots) is allowed, and the wires emanating to the left are allowed to be in the true setting. }
    \label{fig:infty-variable}
\end{figure}

\paragraph*{Turning wires.}
The top wire coming in to a positive clause gadget turns downward as shown  in Figure~\ref{fig:Linf-construction}. 
The false setting along the wire (see (a) in the figure) forces  
the bottom  interval of the ladder to have its false setting on the left (see  point $p_3$ in the figure).  The ladder can be extended to the appropriate length by adding multiples of three intervals.
The true setting is shown in Figure~\ref{fig:Linf-construction}(b), and allows the right endpoint of the clause interval to be used as a representative point. 

The bottom wire coming in to  a clause gadget turns  upward as shown  in Figure~\ref{fig:Linf-construction}. 
The false setting along the wire (see (a) in the figure) forces the top interval of the ladder to have its false setting on the right (see point $p_1$ in the figure). 
The ladder can be extended to the appropriate length by adding multiples of three intervals.
The true setting is shown in Figure~\ref{fig:Linf-construction}(b), and allows the midpoint of the clause interval to be used as a representative point.
We have now established Property~\ref{prop:wire} for wires that turn.

\begin{figure}[htb]
    \centering
    \includegraphics[width=\textwidth]{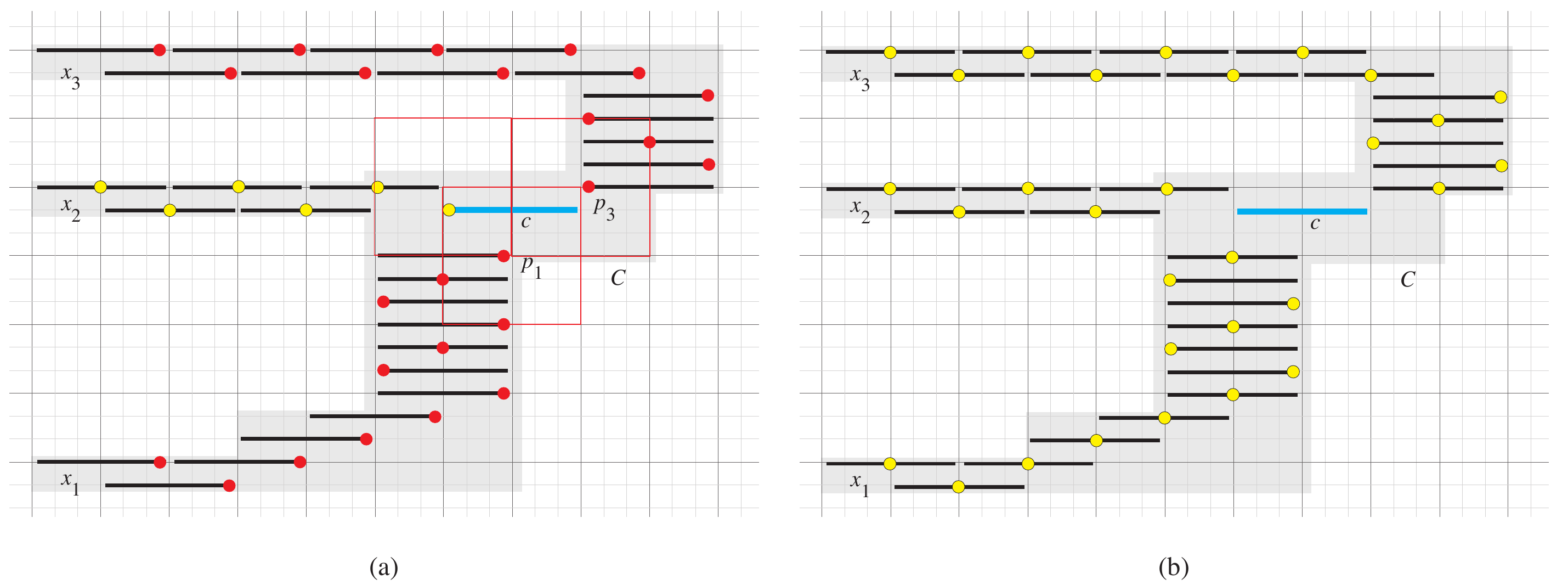}
    \caption{The clause gadget and its entering wires for $L_\infty$
    for clause $C = x_1 \vee x_2 \vee x_3$.  The clause interval $c$ is shown in cyan.
    (a) Valid representative  points  for the truth-value setting $x_1 = $ False, $x_2 = $ True, $x_3 = $ False.  
    (b) The true settings on the incoming wires.
    (These points are not forced.)
    }
    \label{fig:Linf-construction}
\end{figure}



\paragraph*{Clause gadget.}
See Figure~\ref{fig:Linf-construction}.
The figure shows the clause interval $c$ together with the three wires coming in to the clause gadget.  
Let $p_1$, $p_2$, $p_3$ be the  representative points of the last intervals in the wires entering the clause gadget.  Consider the false positions of $p_1, p_2, p_3$. (Figure~\ref{fig:Linf-construction}(a)
shows the false positions of $p_1$ and $p_3$.)
The $L_\infty$ distance between $p_1$ and either endpoint of $c$ is $\delta_\infty = \frac{1}{2}$.
Let $B_i$ be the $L_\infty$ ball of radius $\delta_\infty$ centred at $p_i$, $i=1,2,3$.  
We now verify Property~\ref{prop:clause}.
Observe that no  point of the interval $c$ is outside all three balls. Thus, if all three incoming  wires have the false setting, there is  no valid representative  point for interval $c$.
However, if at least one of the incoming wires has the true setting, we claim that there is a valid representative point on interval $c$.  Refer to Figure~\ref{fig:Linf-construction}(b).
If $p_1$ is at the middle of its interval, use the midpoint of $c$, and if either of $p_2, p_3$ is at the middle of its interval, use the endpoint of $c$ on that side.
Thus  Property~\ref{prop:clause} holds.

\subsection{Bit complexity and containment in NP}
\label{sec:bit-complexity}

In this section we show that the decision version of distant representatives for rectangles is contained in NP for the $L_1$ and $L_\infty$ norms, and discuss why this is open for the $L_2$ norm.

\remove{
In this section we 
address the question: given an input to the distant representatives problem, what is the bit complexity of $\delta^*$ and of representative points realizing $\delta^*$?
In $L_2$ we show that $\delta^*$ may be irrational.  In $L_1$ and $L_\infty$ we give polynomial bounds on the size of $\delta^*$ and representative points. 
Furthermore, for $L_\infty$ we show that there are only $O(n^3)$ possible values for $\delta^*$. 
As a corollary of these polynomial bounds, the decision version of distant representatives lies in NP for the $L_1$ and $L_\infty$ norms. This is an open question for $L_2$.
}
\remove{
\subsubsection{$L_\infty$-norm}

\anna{This stuff was moved up into the optimization section.}

\subsubsection{$L_1$-norm}

While we can find a set of $O(n^3)$ candidates for the value of $\delta_\infty^*$, we do not know of a similar result for $L_1$.  
One can easily generalize Lemma~\ref{lem:tight} and show (for example) that any rectangle $R_j$ is either left-represented or there exists a representative $p_i$ to the left of $p_j$ with $d_\ell(p_i,p_j)=\delta^*$.  But Corollary~\ref{cor:tight} does not hold, since the distance between two points now depends on both $x$-coordinate and $y$-coordinate, and in consequence 
it seems impossible to extract a formula for the optimum value of $\delta^*$ from this.

\anna{Reformulate this (below) to claim poly bound on size of $\delta^*$ and representative points.  State as a corollary that the decision problem lies in NP.  Note that we do not need to \emph{solve} an LP to get our results.}

Using Corollary~\ref{cor:tight} and Lemma~\ref{lem:delta_form}, one can easily argue
that the decision variant of the distant representatives problem is in NP for the
$L_\infty$-norm, because there exists an optimal solution whether the coordinates
are polynomial in the input-coordinates.   

We can prove membership in NP also for the $L_1$-norm, and to this end, 
appeal to a more general argument from Abrahamsen et al.~\cite{abrahamsen2020framework}.
\anna{Our problem doesn't exactly match theirs (they are packing given objects (i.e. $\delta$ is given) and they just attribute to folklore anyway, so I suggest we just mention folklore rather than referring to them.}

\begin{lemma}
The decision version of the distant representatives problem is in NP for the
$L_1$-norm. 
\end{lemma}
\anna{Note that the proof shows something stronger---$\delta^*$ has a poly number of bits.}
\begin{proof}
We must argue that we can describe and verify a solution using polynomially many bits
and polynomial time.  We describe a solution via two permutations $\pi_x$ and $\pi_y$
that are the sorting permutations of an optimum solution with respect to the
$x$-coordinate and $y$-coordinates (with ties broken arbitrarily).  Given these
sorting permutations, the optimal solutions in the $L_1$-norm 
can then easily be described as the solution of a linear program as follows:
$$
\begin{array}{llll}
\max & \delta_1^* & & \\
\text{subject to} 
& x_j-x_i + y_j-y_i \geq \delta_1^* & \text{for all pairs $i,j$ with $\pi_x(j)>\pi_x(i)$ and $\pi_y(j)>\pi_y(i)$} \\
& x_j-x_i + y_i-y_j \geq \delta_1^* & \text{for all pairs $i,j$ with $\pi_x(j)>\pi_x(i)$ and $\pi_y(j)<\pi_y(i)$} \\
& x_i-x_j + y_j-y_i \geq \delta_1^* & \text{for all pairs $i,j$ with $\pi_x(j)<\pi_x(i)$ and $\pi_y(j)>\pi_y(i)$} \\
& x_i-x_j + y_i-y_j \geq \delta_1^* & \text{for all pairs $i,j$ with $\pi_x(j)<\pi_x(i)$ and $\pi_y(j)<\pi_y(i)$}  \\
& x_j-x_i \geq 0 & \text{for all pairs $i,j$ with $\pi_x(j)>\pi_x(i)$} \\
& y_j-y_i \geq 0 & \text{for all pairs $i,j$ with $\pi_y(j)>\pi_y(i)$} \\
& \ell_i \leq x_i \leq r_i & \text{for all $i$} \\
& b_i \leq y_i \leq t_i & \text{for all $i$} 
\end{array}
$$
The optimum solution of a linear program can be found in $O(N+\log D)$ time, where $N$ is
the number of variables and constraints and $D$ is the maximum absolute value of a restriction
For our linear program, we have $N\in O(n^2)$ while
$D$ equals the maximum absolute coordinate, so $\log D$ is polynomial in the
input-size.  So we can verify the solution to the $L_1$-norm decision problem in polynomial
time as desired.
\end{proof}

We note that this linear-program approach could also be used to prove that the decision problem
for the $L_\infty$-norm is in NP; here we would describe the solution by giving the sorting
permutations as well as and indicator for every pair of rectangles whether the distance is
achieved in the $x$-direction or the $y$-direction.
}

\subsubsection{$L_\infty$-norm and $L_1$-norm}

One can immediately argue that the problem lies in NP for the $L_\infty$-norm
due to Lemma~\ref{lemma:n-cubed}.
We can prove membership in NP also for the $L_1$-norm, and to this end, 
prove a more general statement on the bit complexity.

\begin{lemma}
\label{lemma:L1-bit-bound}
The number of bits for $\delta^*$ and the coordinates of an optimal solution is polynomial in $n+\log D$ for the $L_1$-norm.
\end{lemma}
\begin{proof}
Fix an arbitrary optimal solution, and let $\pi_x$ and $\pi_y$ be its sorting permutation with respect to the $x$-coordinate and $y$-coordinate (breaking ties arbitrarily).  Then clearly some optimum solution is a solution to the following linear program (with variables $\delta_1^*$ and $x(R),y(R)$ for each rectangle $R$):
$$
\begin{array}{rlrlll}
\max & \delta_1^* & \\
\multicolumn{2}{r}{\text{subject to} }
& (x(R),y(R))\in R & \text{for all rectangles $R$} \\
&& x(R)\leq x(R')  & \text{for all rectangles $R,R'$ with $R$ before $R'$ in $\pi_x$} \\
&& y(R)\leq y(R')  & \text{for all rectangles $R,R'$ with $R$ before $R'$ in $\pi_y$} \\
\multicolumn{3}{r}{\quad d_1(\, (x(R),y(R)),(x(R'),y(R')\,) \geq \delta_1^*} & \text{for all rectangles $R,R'$}
\end{array}
$$
This is indeed a linear program, since we fixed the relative order of coordinates via $\pi_x$ and $\pi_r$ and hence $d_1(R,R')=|x(R)-x(R')|+|y(R)-y(R')|$ is a linear function in the domain.  Since linear programming is in NP 
we know that there exists an optimum solution where variables and solution-value have polynomial size in the input-numbers.
\end{proof}

\subsubsection{$L_2$-norm}

For the $L_2$-norm  $\delta^*$ may be irrational, even if all coordinates of rectangles are integers.  To see this, consider an input consisting of three identical unit squares (all on top of each other).  Then we must place three points inside a unit square while maximizing their pairwise distances.
The optimum solution is realized by placing an  equilateral triangle inside a square with the base of the triangle rotated $15^\circ$ from horizontal.  This gives $\delta^*=\sqrt 6 - \sqrt 2$ and point coordinates $(0,0), (1, 2 - \sqrt{3}), (2 - \sqrt{3}, 1)$ \url{https://mathworld.wolfram.com/EquilateralTriangle.html}. 
Note that $\delta^2$ is also irrational.

For the decision problem in the $L_2$ norm, 
we do not know if irrational representative points may be necessary if $\delta$ is rational.
Nor do we know
whether the decision problem is in NP---note that the linear constraints in the LP for the $L_1$ case in the proof of Lemma~\ref{lemma:L1-bit-bound} become quadratic constraints.  

\subsection{APX-hardness}
\label{sec:appendix-APX-hardness}

\remove{

In this section, we prove hardness-of-approximation results for the distant representatives problem on  horizontal and vertical segments in the plane.
Specifically, 
we prove lower bounds on the approximation factors that can be achieved in polynomial time, assuming P $\ne$ NP.

\begin{theorem} For $\ell = 1, 2, \infty$, let $g_\ell$ be the constant shown in Table~\ref{Tab:ratios}.  Suppose P $\ne NP$.  Then, for the $L_\ell$ norm,  there is no polynomial time algorithm with approximation factor less than $g_\ell$ for the distant representatives problem for horizontal and vertical segments.
\label{PTAS_claim}
\end{theorem}
\begin{table}[ht]
    \centering
    
    \begin{tabular}{c|c c c}
           & $L_1$ & $L_2$ & $L_\infty$\\
           \hline
           lower bound &  $g_1 = 1.5$ & $g_2 = 1.4425$ 
           & $g_{\infty} = 1.5$ \\
    \end{tabular}
    
    \caption{Best approximation ratios that can be achieved unless P=NP.} 
    \label{Tab:ratios}
\end{table}


We prove Theorem~\ref{PTAS_claim} using a 
\emph{gap reduction}.  
This standard approach is based on the fact that if there were polynomial time approximation algorithms with approximation factors better than $g_\ell$ then the \emph{gap versions} of the problem (as stated below) would be solvable in polynomial time.
Thus, proving that the gap versions are NP-hard implies that there are no polynomial time $g_\ell$-approximation algorithms unless P=NP.

Recall that $\delta^*_\ell$ is the max over all assignments of representative points, of the min distance between two points.

\begin{quote}
{\bf Gap Distant Representatives Problem.}\\
{\bf Input:} A set $I$ of horizontal and vertical segments in the plane.\\
{\bf Output:} 
\begin{itemize}
    \item YES if $\delta^*_\ell(I) \ge 1$;
    \item NO if $\delta^*_\ell(I) \le 1/g_\ell$;
    \item  and it does not matter what the output is for other inputs.
\end{itemize}

\end{quote}

To prove Theorem~\ref{PTAS_claim} it therefore suffices to prove:

\begin{theorem}
The Gap Distant Representatives problem is NP-hard.
\label{theorem:gap-NP-hard}
\end{theorem}


We prove this via a reduction from Monotone Rectilinear Planar 3-SAT, much like in the previous section. 
The gadgets are simpler because we can use vertical segments, but we must prove stronger properties. 
Given an instance $\Phi$ of Monotone Rectilinear Planar 3-SAT
we construct in polynomial time a set of horizontal and vertical segments $I$ 
with the following properties.

\begin{claim}
\label{claim:if-SAT}
If $\Phi$ is satisfiable then $\delta^*_\ell(I) = 1$.
\end{claim}

\begin{claim}
\label{claim:if-not-SAT}
If $\Phi$ is not satisfiable then $\delta^*_\ell(I) \le 1/g_\ell$.
\end{claim}

Thus a polynomial time algorithm for the Gap Distant Representatives problem yields a polynomial time algorithm for Monotone Rectilinear Planar 3-SAT.

We first describe the construction and then prove the claims.

} 

\paragraph*{Further reduction details}

\remove{ 
We reduce directly from Monotone Rectilinear Planar 3-SAT.  

\begin{figure}[htb]
    \centering
    \includegraphics[width=\textwidth]{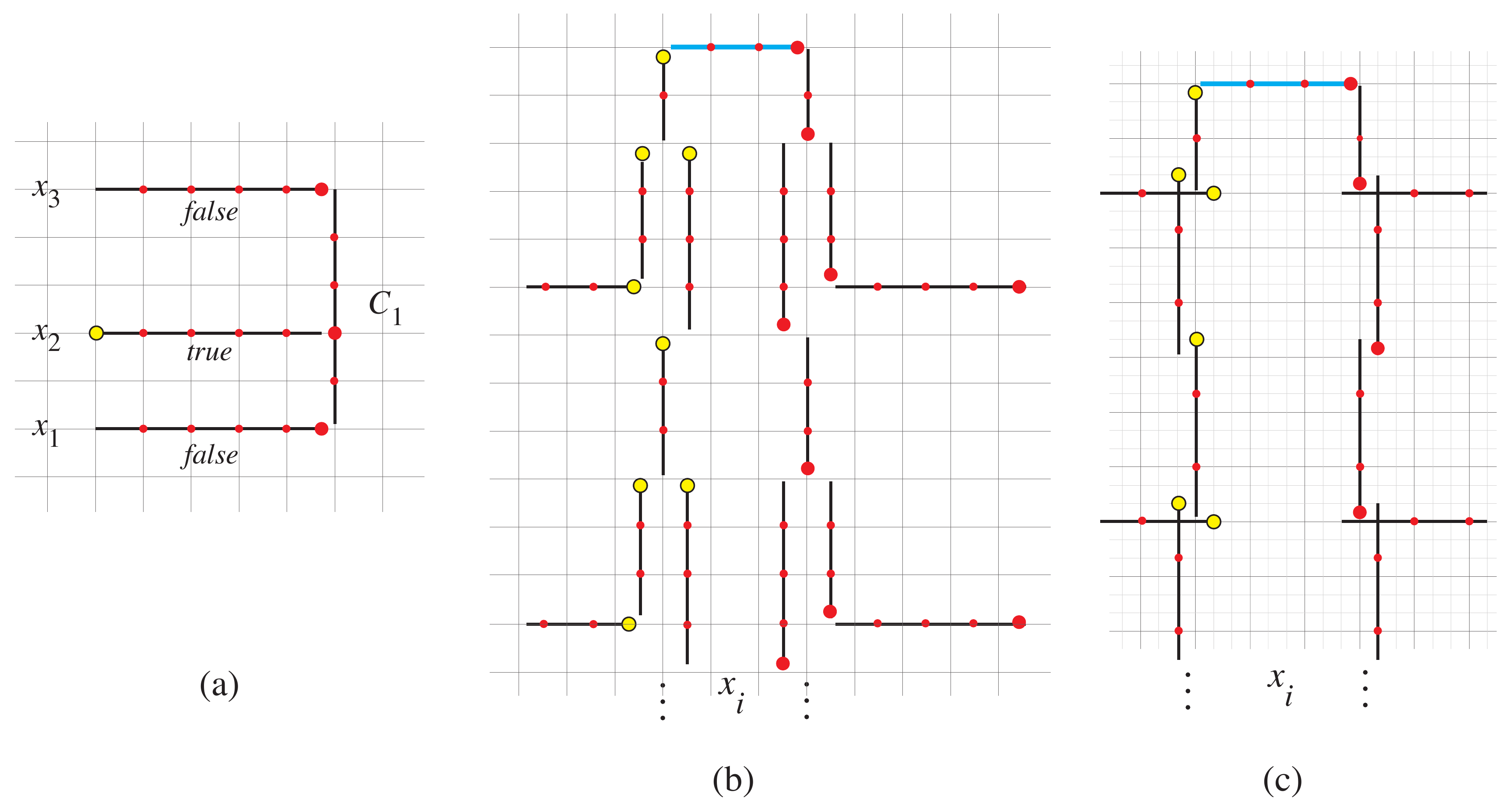}
    \caption{Wire, clause and splitter gadgets.  For clarity, segments are not drawn all the way to their endpoints.
    (a) Three wires attaching to the clause gadget for $C_1$.  Small red dots indicate the representative points on the 0-length segments.
     Wires $x_1$ and $x_2$ are in the false setting and wire $x_2$ is in the true setting, which allows the representative point for $C_1$ to be placed where the $x_2$ wire meets it, while keeping representative points at least distance 1 apart.
    (b) The basic splitter gadget for $L_{\infty}$ on the half grid showing two wires extending left and two right.  The variable segment (in thick cyan) for the variable $x_i$ has its representative point (the large red dot) at the right, which is the false setting.
    The representative points shown by large red/yellow dots are distance at least 1 apart in $L_\infty$.
    (c) The splitter gadget for $L_1$ on a grid subdivided into thirds, with chosen representative points distance at least 1 apart.
     }
    \label{fig:approx-plan}
\end{figure}

\paragraph*{Wire.} A wire is a long horizontal segment 
with 0-length segments at unit distances along it, except at its left and right endpoints. 
See Figure~\ref{fig:approx-plan}(a). 
The representative point for a 0-length segment must be the single point in the segment.  These are shown as small red dots in the figure.
As before, a wire is directed from the variable gadget to the clause gadget.
We distinguish a 
``false setting'' where the wire has its representative point within distance 1 of its forward end (at the clause gadget) and a ``true setting'' where the wire has its representative point within distance 1 of its tail end (at the variable gadget).

\paragraph*{Clause gadget.}  A clause gadget is a vertical segment.  Three wires corresponding to the three literals in the clause meet the vertical segment as shown in Figure~\ref{fig:approx-plan}(a).  There are 0-length segments at unit distance along the clause interval except where the three wires meet it. 
} 

\remove{
A variable segment has length 3, with two 0-length segments placed 1 and 2 units from the endpoints.  A representative point in the right half corresponds to a false value for the variable, and a representative point in the left half corresponds to a true value. 
In order to transmit the variable's value to all the connecting horizontal wires we build a ``splitter'' gadget.  The basic splitter gadget for $L_\infty$ is shown in Figure~\ref{fig:approx-plan}(b).
} 
To obtain our claimed lower bounds, the basic splitter gadget of Figure~\ref{fig:approx-plan}(c) must be modified for the $L_2$ and $L_1$ norms.
For $L_2$ we modify the spacing slightly.  In particular (see Figure~\ref{fig:distances}(c)), we overlap successive vertical segments by a small amount $t$. In order to keep the endpoint of each segment at distance 1 from the nearest 0-length segments, we must have   
$t \le t^* = 1 - \sqrt{3}/2 \approx .13397$ as indicated by the large radius 1 circle in the figure.  To get back on the grid, we increase the gap between the next two 0-length segment to $1+t$.  
In order to have rational coordinates, we use $t = 2/15 = .133\dot{3}$.
(Closer continued fraction approximations to $t^*$ give marginal improvements in $g_2$.)   

For $L_1$ we  construct an alternate splitter gadget with crossing segments as shown in Figure~\ref{fig:L1-splitter}.
As in the $L_2$ case, the gap between a 0-length segment and a segment endpoint or another 0-length segment is increased in some cases, specifically from 1 to $4/3$.

\begin{figure}[htb]
    \centering
    \includegraphics[width=.3\textwidth]{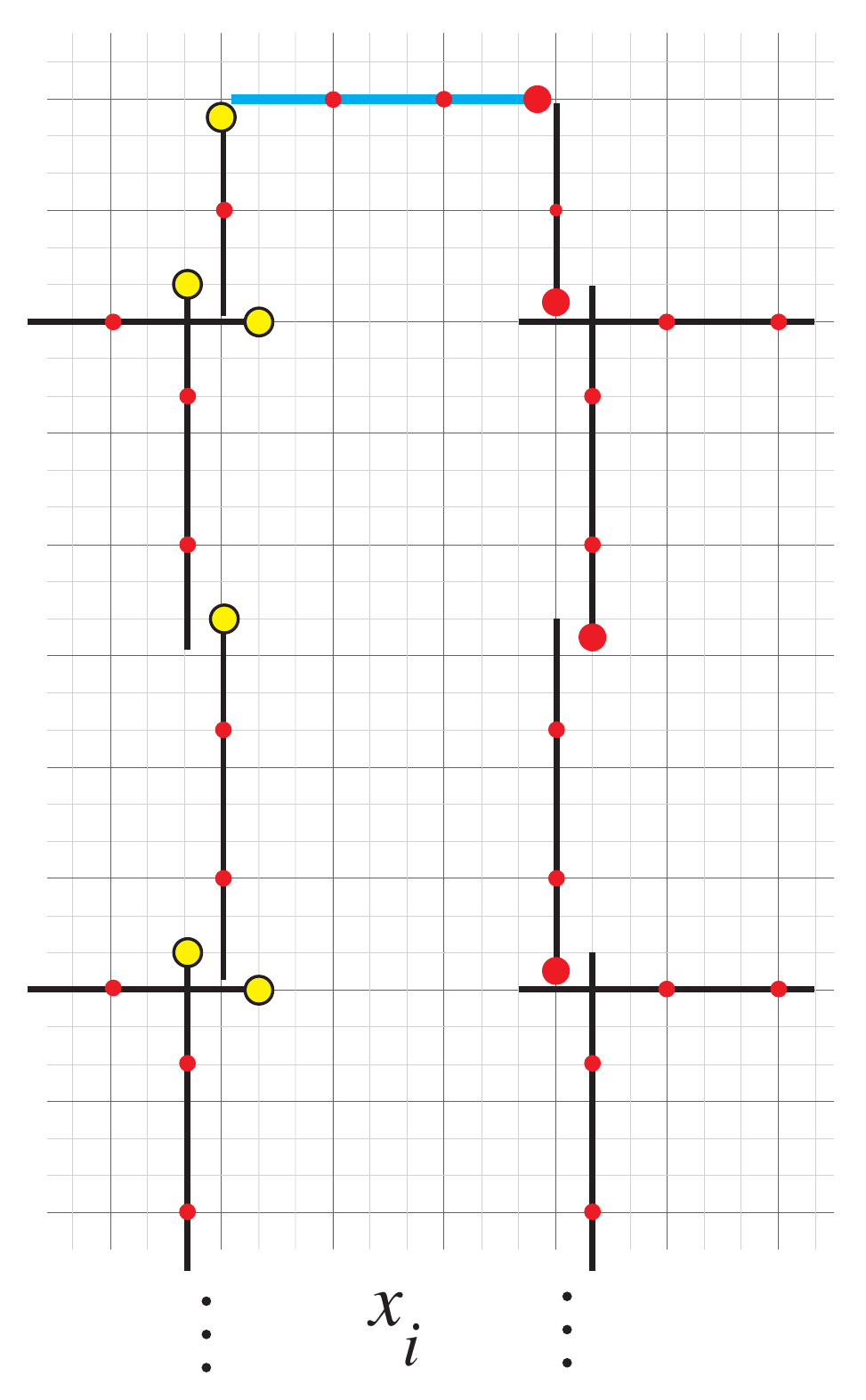}
    \caption{
     The splitter gadget for $L_1$ on a grid subdivided into thirds with 
    two wires extending left and two right.  The variable segment (in thick cyan) for the variable $x_i$ has its representative point (the large red dot) at the right, which is the false setting.
    The representative points shown by large red/yellow dots are distance at least 1 apart in $L_1$.
      }
    \label{fig:L1-splitter}
\end{figure}

\paragraph*{Reduction correctness}

\begin{proof}[Proof of Claim~\ref{claim:if-SAT}]
Suppose the formula $\Phi$ is satisfiable.  We show that there is an assignment of representative points for the intervals in $I$ such that the distance between any two points is at least 1 in the $L_\ell$ norm. 
We cannot do better than distance 1 because there are 0-length segments at distance exactly 1.
Thus $\delta^*_\ell(I) = 1$.

For each variable segment, we place its representative point at the right endpoint if the variable is True in $\Phi$, and at the left endpoint otherwise. 
Representative points for the other segments in the splitter gadget are placed at the endpoints of the segments as shown in Figure~\ref{fig:approx-plan}.
See also Figure~\ref{fig:distances}(c) for details of the $L_2$ case.
For each wire, we place its representative point at the forward end (at the clause gadget) if the corresponding literal is false, and at the tail end (in the splitter gadget) if the corresponding literal is true.  Each clause has a True literal---choose one and  place the representative point for the clause segment at the place where the wire for this True literal meets it.  
See Figure~\ref{fig:approx-plan} which shows that in all cases, the distance between any two representative points is at least 1.
\end{proof}

\begin{proof}[Proof of Claim~\ref{claim:if-not-SAT}]
Suppose the formula $\Phi$ is not satisfiable. Consider any assignment of representative points to the intervals of $I$.  We will show that there are two representative points within distance $1/g_\ell$.  Note that $1/g_1 = 1/g_\infty = 2/3$ and  $1/g_2 = .69324 . . . $. Observe that finding two points within distance $2/3$ suffices for all norms.

The representative points determine a truth value assignment $\cal V$ to the variables as follows: if a variable segment has its representative point in the right half, assign it True, otherwise assign it False.  Since $\Phi$ is not satisfiable, there must be some clause $C$ whose three literals are all false under the assignment $\cal V$. Suppose that $C$ contains three positive literals (the case of three negative literals is symmetric). 

Observe that if the representative point on 
a 
segment
is placed in the unit gap between two  0-length segments, then there are two points within distance $1/2$, which is less than $1/g_\ell$.  In the $L_2$ and $L_1$ splitters, we created longer gaps between successive 0-length segments. 
In the $L_2$ splitter, there are gaps of length $1+t$; a point in such a gap would cause two points to be within distance $(1+t)/2 = .56\dot{6}$ which is less than $1/g_2$.  In the $L_1$ splitter there are gaps of length $4/3$; a point in such a gap would cause two points to be within distance $2/3 = 1/g_1$. 

Thus we may suppose that every wire and every segment in a splitter gadget has its representative point ``near'' (i.e., within distance 1 of) one of its two endpoints, and  
that every clause segment has its representative point near an incoming wire. 
For the clause $C$ (the one whose three literals are all false), suppose its representative point is near incoming wire $w$, and suppose $w$ is associated with variable $x$.
We now separate into two cases depending whether the wire $w$ has its representative point near the forward end or the tail end. 

\begin{figure}[htb]
    \centering
    \includegraphics[width=.9\textwidth]{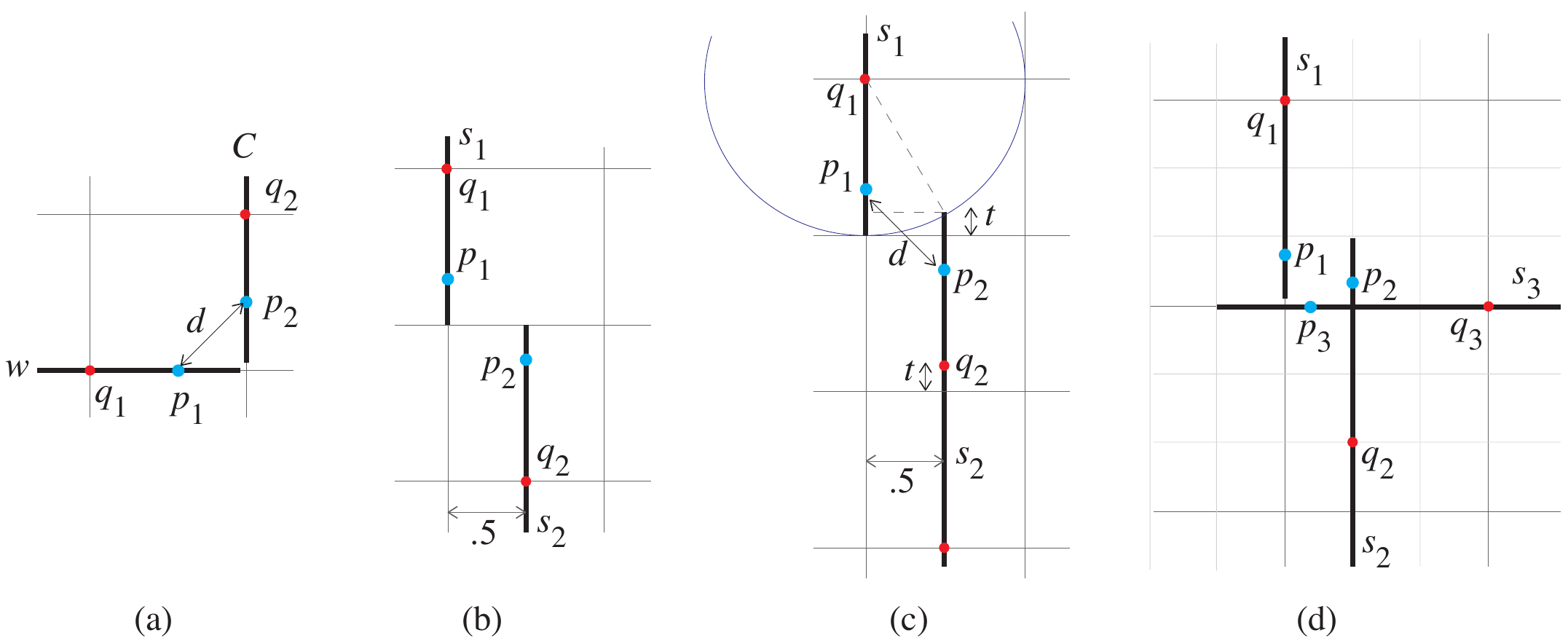}
    \caption{
    (a) A representative point $p_2$ on clause segment $C$ near a representative point $p_1$ on  wire $w$. 
    (b) Representative points $p_1$ and $p_2$ on successive vertical segments of the $L_\infty$ splitter.
    (c) Representative points $p_1$ and $p_2$ on successive segments of the $L_2$ splitter.  In this splitter successive vertical segments overlap by amount $t$. The top endpoint of $s_2$ is not inside the unit circle centred at $q_1$, so long as $t \le t^* = 1 - \sqrt{3}/2$.  
    (d) Representative points $p_1, p_2, p_3$ on successive  segments of the $L_1$ splitter.
    }
    \label{fig:distances}
\end{figure}

\smallskip
\noindent
{\bf Case 1.} The wire $w$ has its representative point near the forward (clause) end.
The situation is as shown in Figure~\ref{fig:distances}(a), with representative point $p_1$ on the wire $w$ and representative point $p_2$ on the clause segment.  If $p_i$ is within distance $2/3$ from its nearest 0-length segment $q_i$, we are done.  Otherwise consider $d_\ell(p_1,p_2)$, marked $d$ in the figure.  We have $d_\ell(p_1, p_2) \le d_1(p_1, p_2) \le 2/3$.

\remove{
There are three distances involved: from $p$ to the closest 0-length segment on the wire, from $p$ to $q$, and from $q$ to the closest 0-length segment on the clause segment.  To maximize the minimum, the three distances should all have the same value $d$.  
We solve for $d$ \graeme{maybe mention that RHS geometrically represents $d_\ell(p,q)$} in each of the norms as follows:

($L_\infty$) $d=1-d$ so $d=1/2$ which is $\le 1/g_\infty = 2/3$. 

\graeme{($L_\infty$) $d=\max(1-d,1-d)=1-d$ so $d=1/2$ which is $\le 1/g_\infty = 2/3$.}

($L_1$) $d=2(1-d)$ so $d=2/3$ which is $ \le 1/g_1 = 2/3$. 

($L_2$) $d = \sqrt{2(1-d)^2}$, so $d^2 - 4d +2=0$ which solves to 
$d= 2-\sqrt 2 \approx .5858$ which is $\le 1/g_2 = .69324 . . . $.

\graeme{Can keep it like this with 3 cases. Or can do one general case, but probably need two or so sentences.}

\noindent
In all cases, we have two points  within distance $1/g_\ell$.
}

\smallskip
\noindent
{\bf Case 2.} The wire $w$ has its representative point near its tail end.  Recall that wire $w$ is associated with variable $x$ which is set False, i.e., the variable segment for $x$ 
has its representative point near the right end.  

In the splitter gadget  
there is a sequence of segments from  the variable segment for $x$ to the wire $w$. 
Somewhere along the sequence there must be two consecutive segments $s_1$ and $s_2$ with representative points $p_1$ and $p_2$ where $p_1$ is near the end  of $s_1$ and $p_2$ is near the start of $s_2$.

If the endpoints of $s_1$ and $s_2$ meet at a right angle, then the analysis in Case 1 shows that there are two points within distance $2/3$.

We separate the remaining cases by the norm.  
For the $L_\infty$ norm  the segments $s_1$ and $s_2$ must be vertical, as shown in 
Figure~\ref{fig:distances}(b).
If either point $p_i$ is within distance $2/3$ of its nearest 0-length segment $q_i$, we are done. 
Otherwise $p_1$ and $p_2$ lie in a rectangle of size $\frac{1}{2} \times \frac{2}{3}$ so their $L_\infty$ distance is at most $2/3$.

Next we consider the $L_1$ norm.  
See Figure~\ref{fig:distances}(d), which shows segment $s_1$ and two possible following segments $s_2$ and $s_3$, together with possible representative points $p_i$ on $s_i$ and the nearest 0-length segment $q_i$ on $s_i$, $i=1,2,3$. If $p_i$ is within distance $2/3$ of $q_i$ we are done.
Otherwise $p_1$ and $p_2$ lie in a square of side-length $1/3$ so their $L_1$ distance is at most $2/3$. 
The same is true for $p_1$ and $p_3$. 

Finally we consider the $L_2$ norm, see points $p_1$ and $p_2$ in 
Figure~\ref{fig:distances}(c).  This is the one case were we do not guarantee two points within distance $2/3$ but only within the higher bound $1/g_2 = .69324 . . . $.  
There are three distances involved: from $p_1$ to the closest 0-length segment $q_1$ on $s_1$, from $p_1$ to $p_2$, and from $p_2$ to the closest 0-length segment $q_2$ on $s_2$. 
To maximize the minimum, the three distances should all have the same value $d$.  

%

Then $d = \sqrt{(\frac{1}{2})^2 + (2-t - 2d)^2}$, so $3d^2 - 4(2-t)d + (2-t)^2 + .25 = 0$. With $t = 2/15$ this solves to 
$d \approx .693237$ which is $\le 1/g_2 = .69324 . . . $. 

This completes the proof of the claim.
\end{proof}

\end{appendix}

\end{document}